\documentclass[3p,review]{elsarticle}

\usepackage{lineno,hyperref,graphicx,epstopdf}
\usepackage{times}
\usepackage{type1cm}
\fontsize{12pt}{18pt}

\usepackage{lineno,hyperref}
\usepackage{amsmath,amsfonts}
\usepackage{algorithmic}
\usepackage{array}
\usepackage[caption=false,font=normalsize,labelfont=sf,textfont=sf]{subfig}
\usepackage{textcomp}
\usepackage{stfloats}
\usepackage{url}
\usepackage{verbatim}
\usepackage{graphicx}
\usepackage{epstopdf}
\usepackage{dcolumn}% Align table columns on decimal point
\usepackage{bm}% bold math
\usepackage{braket}
\usepackage{color}
\usepackage{url}
\usepackage{algorithmic}
\usepackage{algorithm}
\usepackage{balance}
\usepackage{booktabs}%导入表格的宏包
\usepackage{braket} %量子算符宏包 

\allowdisplaybreaks[3]
\modulolinenumbers[5]

\usepackage{braket}
\usepackage{amsmath}
\usepackage{amsfonts}
\usepackage{amssymb}
\usepackage{amsthm}
\usepackage[thicklines]{cancel}

\theoremstyle{plain}
\newtheorem{theorem}{Theorem}

\newtheorem{lemma}{Lemma}

\theoremstyle{definition}
\newtheorem{definition}{Definition}

\makeatletter
\newenvironment{breakablealgorithm}
  {% \begin{breakablealgorithm}
   \begin{center}
     \refstepcounter{algorithm}% New algorithm
     \hrule height.8pt depth0pt \kern2pt% \@fs@pre for \@fs@ruled
     \renewcommand{\caption}[2][\relax]{% Make a new \caption
       {\raggedright\textbf{\ALG@name~\thealgorithm} ##2\par}%
       \ifx\relax##1\relax % #1 is \relax
         \addcontentsline{loa}{algorithm}{\protect\numberline{\thealgorithm}##2}%
       \else % #1 is not \relax
         \addcontentsline{loa}{algorithm}{\protect\numberline{\thealgorithm}##1}%
       \fi
       \kern2pt\hrule\kern2pt
     }
  }{% \end{breakablealgorithm}
     \kern2pt\hrule\relax% \@fs@post for \@fs@ruled
   \end{center}
  }
\makeatother

\journal{}

\date{}

\begin{document}
\begin{frontmatter}
\title{Distributed Exact Generalized Grover's Algorithm}

\author[mymainaddress,myfourthaddress]{Xu Zhou\corref{mycorrespondingauthor}}
\ead{zhoux359@mail.sysu.edu.cn}
\cortext[mycorrespondingauthor]{Corresponding author}

\author[mysecondaryaddress]{Xusheng Xu}
\ead{thuxuxs@163.com}

\author[mythirdaddress]{Shenggen Zheng}
\ead{zhengshenggen@quantumsc.cn}

%\author[myfourthaddress,myfifthaddress,mysixthaddress]{Daowen Qiu}
%\ead{issqdw@mail.sysu.edu.cn}

\author[mymainaddress,mythirdaddress,myfourthaddress]{Le Luo}
\ead{luole5@mail.sysu.edu.cn}

\address[mymainaddress]{School of Physics and Astronomy, Sun Yat-sen University, Zhuhai 519082, China}

%\address[mysecondaryaddress]{Central Research Institute, 2012 Labs, Huawei Technologies, Shenzhen, 518129, China}

\address[mysecondaryaddress]{Department of Physics, State Key Laboratory of Low-Dimensional Quantum Physics, Tsinghua University, Beijing 100084, China}

\address[mythirdaddress]{Quantum Science Center of Guangdong-Hong Kong-Macao Greater Bay Area (Guangdong), Shenzhen 518045, China}

%\address[myfourthaddress]{Institute of Quantum Computing and Computer Theory, School of Computer Science and Engineering, Sun Yat-sen University, Guangzhou, 510006, China}
%
%\address[myfifthaddress]{The Guangdong Key Laboratory of Information Security Technology, Sun Yat-sen University, Guangzhou, 510006, China}

\address[myfourthaddress]{QUDOOR Co, Ltd., Beijing 100089, China}
%\address[mysixthaddress]

\begin{abstract}
Distributed quantum computation has garnered immense attention in the noisy intermediate-scale quantum (NISQ) era, where each computational node necessitates fewer qubits and quantum gates. In this paper, we focus on a generalized search problem involving multiple targets within an unordered database and propose a Distributed Exact Generalized Grover’s Algorithm (DEGGA) to address this challenge by decomposing it into arbitrary $t$ components, where $2 \leq t \leq n$.
%In this paper, we focus on a generalized search problem that comprises multiple targets within an unordered database. We subsequent propose a distributed exact generalized Grover's algorithm (DEGGA) that addresses the initial generalized search difficulty by decomposing it into arbitrary $t$ components, where $2 \leq t \leq n$. 
Specifically, (1) our algorithm ensures accuracy, with a theoretical probability of identifying the target states at $100\%$; 
%(1) our algorithm is accurate, implying that the theoretical chance for identifying the objective states is $100\%$; 
(2) if the number of targets is fixed, the pivotal factor influencing the circuit depth of DEGGA is the partitioning strategy, rather than the magnitude of $n$; %Conversely, the modified Grover's algorithm augments the circuit depth as $n$ escalates; 
(3) our method requires a total of $n$ qubits, eliminating the need for auxiliary qubits;
%(3) the total number of qubits in our method is $n$, indicating that no auxiliary qubits are necessary; 
(4) we elucidate the resolutions (two-node and three-node) of a particular generalized search issue incorporating two goal strings (000000 and 111111) by applying DEGGA. The feasibility and effectiveness of our suggested approach is further demonstrated by executing the quantum circuits on MindSpore Quantum (a quantum simulation software). Eventually, through the decomposition of multi-qubit gates, DEGGA diminishes the utilization of quantum gates by $90.7\%$ and decreases the circuit depth by $91.3\%$ in comparison to the modified Grover's algorithm by Long. It is increasingly evident that distributed quantum algorithms offer augmented practicality. %in the NISQ era.
%It becomes apparent that distributed quantum algorithms exhibit enhanced practicality in the NISQ era.
%Through the decomposition of multi-qubit gates in quantum circuits,
%Besides, the optimal decomposition is the average decomposition.
%(2) if the number of targets is fixed, the key element in determining the circuit depth of DEGGA is the division strategy, not the size of $n$. However, the modified Grover's algorithm deepens the circuit depth as $n$ increases.
\end{abstract}


\end{frontmatter}

%\linenumbers

\section{Introduction}\label{introduction}
Quantum computation is an emerging field of study that leverages the principles of quantum mechanics to develop a new paradigm for computation. The first ideas about quantum computing were proposed by Benioff \cite{RefBenioff1980} and  Feynman \cite{RefFeynman1982} in the 1980s. These early ideas were based on the concept of using quantum mechanical phenomena, such as superposition and entanglement, to perform computation. In 1985, Deutsch \cite{RefDeutsch1985} presented the idea of a universal quantum computer, which could perform any computation that a classical computer could. This idea laid the groundwork for the development of more advanced quantum algorithms. 

Subsequently, the development of quantum algorithms has flourished, with researchers exploring new algorithms for a wide range of tasks. Some of the most important quantum algorithms developed include Deutsch's algorithm \cite{RefDeutsch1985}, Deutsch-Jozsa (DJ) algorithm \cite{RefDeutsch1992}, Bernstein-Vazirani (BV) algorithm \cite{RefBernstein1993}, Simon's algorithm \cite{RefSimon1997}, Shor's algorithm \cite{RefShor1994}, Grover's algorithm \cite{RefGrover1997}, HHL algorithm \cite{RefHarrow2009} and so on. In recent years, quantum-classical hybrid algorithms such as variational quantum algorithm (VQA)  \cite{RefVQA2021}, variational quantum eigensolver (VQE) \cite{RefVQE2014}, and quantum approximate optimization algorithm (QAOA) \cite{RefQAOA2014} have garnered considerable attention. Quantum algorithms have demonstrated unparalleled advantages over classical counterparts for certain problems, owing to the distinctive characteristics of quantum superposition and entanglement.

Grover's algorithm, a groundbreaking quantum search technique, was first introduced by Grover \cite{RefGrover1997} in 1996. This innovative algorithm effectively addresses the search problem in an unordered database containing $N=2^n$ elements. Here we consider a Boolean function $f: \{0,1\}^n \rightarrow \{0,1\}$. In general, the $m$ targets search problem is also known as the generalized search problem, which aims to identify a target string $x^{(i)}\in\{0,1\}^n$ satisfying $f \left(x^{(i)}\right)=1$, where $i\in\{0, 1, \cdots, m-1\}$. By assuming the availability of an Oracle capable of recognizing inputs, it would return $f (x)=1$ if $x\in\{0,1\}^n$ is a target string and $f (x)=0$ otherwise. Classical algorithms demand $\mathcal{O} \left(N/m\right)$ queries to the Oracle, whereas Grover's algorithm achieves square acceleration by only requiring $\mathcal{O} \left(\sqrt{N/m}\right)$ queries, with high probability, to pinpoint the target strings.

Grover's algorithm is widely acknowledged as one of the most pivotal quantum algorithms. It has found extensive applications in various domains, including the minimum search problem \cite{RefChristoph1996}, string matching problem \cite{RefRamesh2003}, quantum dynamic programming \cite{RefAmbainis2019}, and computational geometry problem \cite{RefAndris2020}. Furthermore, an extension of Grover's method, known as the quantum amplitude amplification and estimation algorithms \cite{RefBrassard2002}, has been proposed.

However, the accuracy of Grover's algorithm in searching for the target state is limited, as it has been rigorously mathematically demonstrated \cite{Refiao2010} that it can only achieve precise results when $1/4$ of the data in the database satisfy the search conditions. In an effort to achieve precision in search, three distinct methodologies \cite{RefBrassard2002, RefHöyer2000, RefLong2001} were proposed around 2000, each of which is grounded in the extension of the original Grover's algorithm, albeit with varying conceptual frameworks. Specifically, in 2001, Long \cite{RefLong2001} advanced Grover's algorithm by extending the Grover operator $G$ to an operator $L$, thereby facilitating exact search. The modified version achieved a probability of $100 \%$ in retrieving target states.

In 2017, Preskill \cite{RefPreskill2018} declared that quantum computation is entering the noisy intermediate-scale quantum (NISQ) era. In the pursuit of realizing quantum computers, a multitude of physical systems are currently being investigated, encompassing ions \cite{RefCirac1995}, photons \cite{RefLu2007}, superconduction \cite{RefMakhlin2001}, and other quantum-based systems \cite{RefBerezovsky2008, RefHanson2008}. Currently, the construction of small-qubit quantum computers appears to be more facile compared to that of large-scale universal quantum computers. Consequently, researchers are contemplating the potential for multiple small-scale devices to collaborate and achieve a task on a grand scale. This embodies the essence of distributed quantum computing.

Distributed quantum computation, which merges distributed computing and quantum computation, has garnered immense attention in the current NISQ era. In order to efficiently tackle a colossal issue, it endeavors to decompose it into numerous subproblems that are distributed several quantum computers. Its characteristic is that each computing node requires fewer qubits and possesses a shallower quantum circuit.

This research field has witnessed a substantial amount of theoretical and experimental endeavors \cite{RefBuhrman2003, RefYimsiriwattana2004, RefBeals2013, RefLi2017}. 

Avron et al. \cite{RefAvron2021} proposed a distributed Oracle approach and illustrated its benefits in the experiment of implementing 3-qubit Grover's algorithm. %for the Grover's algorithm. In the experiment of implementing the Grover's algorithm, however, they were only able to obtain $n-1$ qubits of the target state, rather than accessing the target state directly. The total number of qubits employed in their proposed scheme amounts to $2(n-1)$.
Qiu et al. \cite{RefQiu2022} elucidated the serial and parallel distributed variants of Grover's algorithm. Although they successfully addressed the multi-target search conundrum, the required number of qubits is $2^k(n-k)$, with $2^k$ denoting the number of computational nodes. Zhou and Qiu et al.  \cite{RefZhou2023DEGA} devised a distributed exact Grover's algorithm (DEGA), yet their scheme is restricted to resolving scenarios involving a solitary target string.

Building upon Avron's foundational work, Tan, Xiao, and Qiu et al. \cite{RefTan2022} proposed a refined distributed Simon's algorithm with reduced query complexity. However, their methodology necessitates a greater number of qubits compared to the original circuits and demands supplementary classical queries. In terms of generalizability and exactness, Li and Qiu et al. \cite{RefLi2023Simon} provided a distributed quantum algorithm for Simon's problem. 

Zhou and Qiu et al. \cite{RefZhou2023DBVA} created a $t$-node distributed Bernstein-Vazirani algorithm. In the experimental verification, their technique is more straightforward than the original BV algorithm. For dealing with the DJ problem, Li and Qiu et al. \cite{RefLi2023} stated a multi-node distributed exact quantum algorithm. Recently, a distributed Shor's algorithm was suggested by Xiao and Qiu et al. \cite{RefXiao2022}, which enables separately predict patrial bits of $s/r$ for some $s\in\{0,1,\cdots,r-1\}$ by two quantum computers while resolving order-finding. They \cite{RefXiao2023} subsequently broadened their plan to include several nodes. 

%In this paper, we concentrate on a generalized search problem that comprises multiple targets within an unordered database. We subsequent propose a distributed exact generalized Grover's algorithm (DEGGA) that addresses the initial generalized search difficulty by decomposing it into arbitrary $t$ components, where $2 \leq t \leq n$. More specifically, (1) our algorithm is accurate, implying that the theoretical chance for identifying the objective states is $100\%$; (2) provided that the number of targets remains constant, the pivotal factor influencing the circuit depth of DEGGA is the partitioning strategy, rather than the magnitude of $n$. Conversely, the modified Grover's algorithm augments the circuit depth as $n$ escalates; (3) the total number of qubits in our method is $n$, indicating that no auxiliary qubits are necessary.

In this paper, we focus on a generalized search problem involving multiple targets within an unordered database and propose a Distributed Exact Generalized Grover’s Algorithm (DEGGA) to address this challenge by decomposing it into arbitrary $t$ components, where $2 \leq t \leq n$. Specifically, (1) our algorithm ensures accuracy, with a theoretical probability of identifying the target states at $100\%$; (2) if the number of targets is fixed, the pivotal factor influencing the circuit depth of DEGGA is the partitioning strategy, rather than the magnitude of $n$; (3) our method requires a total of $n$ qubits, eliminating the need for auxiliary qubits.

MindSpore Quantum \cite{Refmq_2021}, a quantum simulation software, serves as a comprehensive software library facilitating the development of applications for quantum computation. Furthermore, we expound upon the resolution of a specific generalized search problem involving two goal strings (000000 and 111111) through the implementation of two-node and three-node DEGGA algorithms. %These experiments will elucidate the intricate procedures involved in implementing a 6-qubit DEGGA. 
The viability and effectiveness of our suggested methodology can be further substantiated by executing the quantum circuits on MindSpore Quantum. Eventually, through the decomposition of multi-qubit gates in quantum circuits, it becomes apparent that distributed quantum algorithms exhibit enhanced practicality in the NISQ era. 

The remaining sections of this paper are structured as follows. In Section \ref{preliminary}, we will provide a concise overview of both the original and modified Grover's algorithms. Subsequently, in Section \ref{DEGGA}, our primary focus will be on introducing the DEGGA. Next, we will delve into the in-depth analysis of our algorithm in Section \ref{analysis}. Following this, Section \ref{experiment} will elucidate specific instances of the DEGGA implementation on MindSpore Quantum. Eventually, we will give a succinct  summary in Section \ref{conclusion}.

\section{Preliminary}\label{preliminary}
In this section, we provide a concise overview of the original \cite{RefGrover1997} and modified \cite{RefLong2001} Grover's algorithms, focusing on their applications in generalized search problems.

\subsection{The original Grover's algorithm \cite{RefGrover1997}}\label{sectiongrover}
Let Boolean function $f: \{0,1\}^n \rightarrow \{0,1\}$. First of all, we define the following notation:
\begin{eqnarray}
A&=&\{x\in\{0 , 1\}^n | f(x)=1\},\\
B&=&\{x\in\{0 , 1\}^n | f(x)=0\} ,\\
a &=& |A|=\left| \{x\in\{0,1\}^n|f(x)=1\}\right| \ge 1,\\
b&=& |B|=\left| \{x\in\{0,1\}^n|f(x)=0\}\right|,\\
|A\rangle&=&\frac{1}{\sqrt{a}}\sum\limits_{x\in A}|x\rangle,\\
|B\rangle&=&\frac{1}{\sqrt{b}}\sum\limits_{x\in B}|x\rangle.
\end{eqnarray}
Next, we define a Boolean function for the generalized search problem as follows:
\begin{eqnarray}\label{groverdef}
f(x)=
\begin{cases}
1,x\in A,\\
0,x\in B ,
\end{cases}
\end{eqnarray}
where $x\in\{0,1\}^n$. 

We assume that there exists a black box (Oracle  $U_{f(x)}:\vert x\rangle\rightarrow (-1)^{f (x)}\vert x\rangle$, where $x\in\{0,1\}^n$) that recognises the inputs. The aim of Grover's algorithm is to use this black box to find out $x\in A$ with a high degree of probability.

Here, we briefly review Grover's algorithm in Algorithm \ref{Grover's Algorithm}. The whole quantum circuit is shown in Figure \ref{grovercircuit}.

\begin{breakablealgorithm}
\caption{Grover's Algorithm}
\label{Grover's Algorithm}
\begin{algorithmic}
\noindent \textbf{Input}: (1) The number of qubits $n$; (2) $N=2^n$; (3) A function $f(x)$, where $f(x)=0$ for $x\in\{0,1\}^n$ except $x\in A$, for which $f(x)=1$; (4) The number of target strings $a\ge 1$, which is a smaller number; (5) An Oracle $U_{f(x)}:\vert x\rangle\rightarrow (-1)^{f (x)}\vert x\rangle$, where $x\in\{0,1\}^n$.

\noindent \textbf{Output}: The target string $x\in A$ with great probability.

\noindent\textbf{Procedure:}

\noindent \textbf{Step 1.} Initialize $n$ qubits $\vert 0\rangle^{\otimes n}$, and get 
\begin{eqnarray}
|\psi_0\rangle = |0\rangle^{\otimes n}.
\end{eqnarray} 

\noindent \textbf{Step 2.} Apply $n$ Hadamard gates $H^{\otimes n}$, and get 
\begin{eqnarray}
|\psi_1\rangle = H^{\otimes n} |\psi_0\rangle =\frac{1}{\sqrt{N}}\sum\limits_{x\in\{0 , 1\}^n}|x\rangle.
\end{eqnarray}

\noindent \textbf{Step 3.} Execute Grover operator (denoted as $G$)  $\left\lfloor \frac{\pi}{4}\sqrt{\frac{N}{a}} \right\rfloor$ times, and get 
\begin{eqnarray}
|\psi_2\rangle = G^{\left\lfloor \frac{\pi}{4}\sqrt{\frac{N}{a}} \right\rfloor} |\psi_1\rangle,
\end{eqnarray}
where
\begin{eqnarray}
G&=&-H^{\otimes n} U_0 H^{\otimes n}U_{f(x)}\\
  &=& -\left(I^{\otimes n} - 2\vert \psi_1\rangle\langle \psi_1\vert \right) \left(I^{\otimes n} - 2\vert A\rangle\langle A\vert \right),
\end{eqnarray}
and
\begin{eqnarray}\label{U0equation}
U_0&=& I^{\otimes n}-2(\vert 0\rangle\langle0\vert )^{\otimes n} ,\\
U_{f(x)}&=& I^{\otimes n} - 2\vert A \rangle\langle A \vert.
\end{eqnarray}

\noindent \textbf{Step 4.} Measure each qubit in the basis $\{\vert 0\rangle, \vert 1\rangle\}$.
\end{algorithmic}
\end{breakablealgorithm}

\begin{figure}[H]
    \centering
\includegraphics[width=0.6\textwidth]{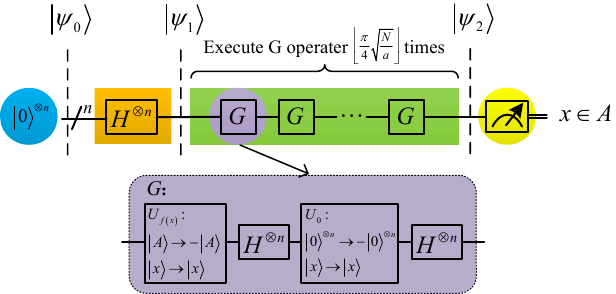}
\caption{Quantum circuit of Grover’s algorithm.}
\label{grovercircuit}
\end{figure}

We provide a briefly analyse of the Grover's algorithm below.

After Step 2 in Grover's algorithm, we have 
\begin{eqnarray}
| \psi_1 \rangle &=& H^{\otimes n} |\psi_0\rangle  = \frac{1}{\sqrt{N}}\sum\limits_{x\in\{0 , 1\}^n}|x\rangle \\
&\triangleq& \sqrt{\frac{a}{N}}|A\rangle+\sqrt{\frac{b}{N}}|B\rangle.
\end{eqnarray}

After $k$ iterations of $G$, we have 
\begin{eqnarray}
|\psi_2\rangle &=& G^{ k} |\psi_1\rangle \\
&= & \cos\left((2k+1)\theta\right)\vert B\rangle + \sin\left((2k+1)\theta\right)\vert A\rangle,
\end{eqnarray}
where 
\begin{eqnarray}
\sin \theta = \sqrt{\frac{a}{N}}, \cos \theta = \sqrt{\frac{b}{N}}.
\end{eqnarray}
To make
\begin{eqnarray}
\sin((2k+1)\theta) \approx 1,
\end{eqnarray}
which means 
\begin{eqnarray}
(2k+1)\theta \approx \frac{\pi}{2},
\end{eqnarray}
so that we can decide
\begin{eqnarray}
k \approx \frac{\pi/2-\theta}{2\theta}=\frac{\pi}{4\theta}-\frac{1}{2}.
\end{eqnarray}
It should be noted that $k$ must be a positive integer.

Due to the presence of  $a$ target items, then
\begin{eqnarray}
\theta = \arcsin \left(\sqrt{\frac{a}{N}}\right) \approx \sqrt{\frac{a}{N}} ~ (N \rightarrow + \infty ).
\end{eqnarray}
In other words, 
\begin{eqnarray}
k = \left\lfloor \frac{\pi}{4}\sqrt{\frac{N}{a}} \right\rfloor
\end{eqnarray}
is an appropriate selection for Grover's algorithm. The success probability is
\begin{eqnarray}
\sin^2 \left( \left(2\left\lfloor \frac{\pi}{4}\sqrt{\frac{N}{a}} \right\rfloor+1\right) \arcsin \left(\sqrt{\frac{a}{N}}\right) \right) \approx 1.
\end{eqnarray}

\subsection{The modified Grover's algorithm \cite{RefLong2001} by Long}\label{sectionlong}
In 2001, Long enhanced Grover's algorithm and introduced a modified version that guarantees the acquisition of the goal states with a probability of $100 \%$. The key concept involves the extension of the Grover operator $G$ to the operator $L$, which is achieved through a three-step adjustment process: (1) replace phase inversion with phase rotation; (2) make the phase rotation angles of the two reflections in $L$ are equal, which is referred to as the phase-matching condition; (3) iterate $J+1$ times for operator $L$, where $J=\lfloor(\pi/2-\theta)/(2\theta)\rfloor$ and $\theta= \arcsin \left({\sqrt{a/N}}\right)$. In fact, $J$ can be any positive integer equal to or greater than $\lfloor(\pi/2-\theta)/(2\theta)\rfloor$.

Here, we briefly review the modified Grover's algorithm by Long in Algorithm \ref{Long's Algorithm}. The whole quantum circuit is shown in Figure \ref{longcircuit}.

\begin{breakablealgorithm}
\caption{The modified Grover's algorithm by Long}
\label{Long's Algorithm}
\begin{algorithmic}
\noindent \textbf{Input}: (1) The number of qubits $n$; (2) $N=2^n$; (3) A function $f(x)$, where $f(x)=0$ for $x\in\{0,1\}^n$ except $x\in A$, for which $f(x)=1$; (4) The number of target strings $a\ge 1$, which is a smaller number; (5) An Oracle $R_{f(x)}:\vert x\rangle\rightarrow e^{i\phi \cdot f (x)}\vert x\rangle$, where $x\in\{0,1\}^n$, $\phi=2\arcsin\left(\sin\left(\frac{\pi}{4J+6}\right) / \sin \theta\right)$, $J=\lfloor(\pi/2-\theta)/(2\theta)\rfloor$, and $\theta= \arcsin \left({\sqrt{a/N}}\right)$.

\noindent\textbf{Output}: The target string $x\in A$ with a probability of $100 \%$.

\noindent \textbf{Procedure:}

\noindent \textbf{Step 1.} Initialize $n$ qubits $\vert 0\rangle^{\otimes n}$, and get 
\begin{eqnarray}
|\psi_0\rangle = |0\rangle^{\otimes n}.
\end{eqnarray} 

\noindent \textbf{Step 2.} Apply $n$ Hadamard gates $H^{\otimes n}$, and get 
\begin{eqnarray}
|\psi_1\rangle = H^{\otimes n} |\psi_0\rangle =\frac{1}{\sqrt{N}}\sum\limits_{x\in\{0 , 1\}^n}|x\rangle.
\end{eqnarray}

\noindent \textbf{Step 3.} Execute modified operator (denoted as $L$) $J+1$ times, and get 
\begin{eqnarray}
|\psi_2\rangle = L^{J+1} |\psi_1\rangle,
\end{eqnarray}
where
\begin{eqnarray}
L&=& -H^{\otimes n} R_0 H^{\otimes n}R_{f(x)}\\
  &=& -\left(I^{\otimes n} - \left(e^{i\phi}-1\right)\vert \psi_1\rangle\langle \psi_1\vert \right) \left(I^{\otimes n} + \left(e^{i\phi}-1\right)\vert A \rangle\langle A \vert  \right),
\end{eqnarray}
and
\begin{eqnarray}\label{U0equation}
R_0 &=&  I^{\otimes n} + \left(e^{i\phi}-1\right)(\vert 0\rangle\langle0\vert )^{\otimes n},\\
R_{f(x)} &=& I^{\otimes n} + \left(e^{i\phi}-1\right)\vert A \rangle\langle A \vert.
\end{eqnarray}

\noindent \textbf{Step 4.} Measure each qubit in the basis $\{\vert 0\rangle, \vert 1\rangle\}$.
\end{algorithmic}
\end{breakablealgorithm}

\begin{figure}[H]
    \centering
\includegraphics[width=0.6\textwidth]{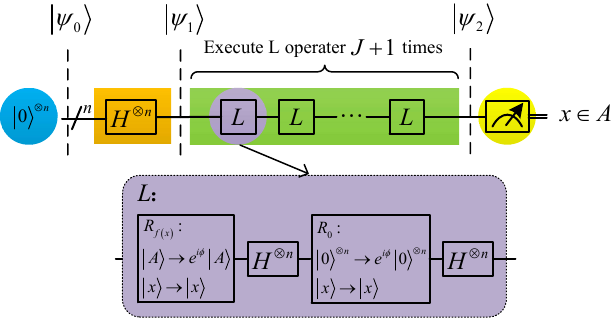}
\caption{Quantum circuit of the modified Grover's algorithm by Long.}
\label{longcircuit}
\end{figure}

We offer a concise analysis of the modified Grover's algorithm by Long.

In fact, $L$ has the following equivalent transformation
\begin{eqnarray}
L&=& -H^{\otimes n} R_0 H^{\otimes n}R_{f(x)}\\
  &=& -\left(I^{\otimes n} - \left(e^{i\phi}-1\right)\vert \psi_1\rangle\langle \psi_1\vert \right) \left(I^{\otimes n} + \left(e^{i\phi}-1\right)\vert A \rangle\langle A \vert  \right),\\
&=& -e^{i\phi} \left[ \cos\left( \frac{\alpha}{2} \right)I + i\sin\left( \frac{\alpha}{2} \right) (n_x X + n_y Y + n_z Z)  \right] \label{laxis},
\end{eqnarray}
where
\begin{eqnarray}\label{alpha}
\alpha = 4\beta, \sin \beta = \sin \left(\frac{\phi}{2}\right) \sin \theta, \theta= \arcsin \left({\sqrt{a/N}}\right),
\end{eqnarray}
and
\begin{eqnarray}\label{nxnynz}
n_x = \frac{\cos \theta}{\cos \beta} \cos \left(\frac{\phi}{2}\right), n_y = \frac{\cos \theta}{\cos \beta} \sin \left(\frac{\phi}{2}\right), n_z = \frac{\cos \theta}{\cos \beta} \cos \left(\frac{\phi}{2}\right) \tan \theta,
\end{eqnarray}
and $I$, $X$, $Y$, $Z$ represent the Pauli gates.

In the three-dimensional space spanned by $\{\vert A\rangle, \vert B\rangle\}$, one iteration of $L$ can be regarded as a rotation by an angle $\alpha$ around the axis $\vec{l}$ from the initial state $| \psi_1 \rangle$ to the target state $| \psi_2 \rangle$, and the overall rotation angle is denoted by $\omega$ (see Figure \ref{rotationL}).
where
\begin{eqnarray}\label{axis}
\vec{l} = \left[n_x, n_y, n_z \right]^T.
\end{eqnarray}
Besides, $\omega$ satisfies
\begin{eqnarray}\label{omega}
\cos \omega  = \cos \left( 2 \arccos \left( \sin \left(\frac{\phi}{2}\right) \sin \theta  \right) \right).
\end{eqnarray}

\begin{figure}[H]
\centering
\includegraphics[width=0.35\textwidth]{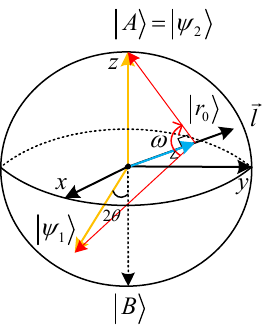}
\caption{Visualization of the modified Grover's algorithm by Long. }
\label{rotationL} 
\end{figure}

Furthermore, in order to be guaranteed that the desired state $| \psi_2 \rangle$ will be obtained with a theoretical success probability of $100\%$, angle $\omega$ should fulfill
\begin{eqnarray}\label{omega2}
\omega  =  (J+1)\alpha = 4 (J+1) \arcsin \left( \sin \left(\frac{\phi}{2}\right) \sin \theta \right)
\end{eqnarray}
after $J+1$ times iteration of $L$. Combining Eq.~\eqref{omega} and Eq.~\eqref{omega2}, we get
\begin{eqnarray}\label{condition}
\sin \left( \frac{\pi}{4J+6} \right) = \sin \left(\frac{\phi}{2}\right) \sin \theta.
\end{eqnarray}
In other words, we have
\begin{eqnarray}\label{phi}
\phi=2\arcsin\left(\sin\left(\frac{\pi}{4J+6}\right) / \sin \theta\right).
\end{eqnarray}

\section{Distributed Exact Generalized Grover's Algorithm}\label{DEGGA}
In this section, we primarily describe the Distributed Exact Generalized Grover's Algorithm (DEGGA) with arbitrary $t$ computing nodes, where $2 \leq t \leq n$. 

\subsection{Subfunctions}\label{subfunctions}
As a conventional approach to basis decomposition, a Boolean function $f: \{0,1\}^n \rightarrow \{0,1\}$ can be factored into two subfunctions:  $f_{\text{even}}(u)=f(u0)$ and $f_{\text{odd}}(u)=f(u1)$ \cite{RefAvron2021},  where $u\in\{0,1\}^{n-1}$. In general, if the $i$-th bit in $x\in\{0,1\}^{n}$ is selected, the following two subfunctions $f_{\text{even}/\text{odd}}: \{0,1\}^{n-1} \rightarrow \{0,1\}$ can be obtained,
\begin{eqnarray}
f_{\text{even}}(v)&=&f(v_0v_1\cdots v_{i-2}0v_{i}\cdots v_{n-2}v_{n-1}),  \label{subfunctions_pra1}\\
f_{\text{odd}}(v)&=&f(v_0v_1\cdots v_{i-2}1v_{i}\cdots v_{n-2}v_{n-1}), \label{subfunctions_pra2}
\end{eqnarray}
where $v=v_0v_1\cdots v_{i-2}v_{i}\cdots v_{n-2}v_{n-1}\in\{0,1\}^{n-1}$.

Without losing generality, if the last $k\in\{2,3,\cdots,n-1\}$ bits in $x\in\{0,1\}^{n}$ are picked to divide $f$, the following $2^k$ subfunctions $f_p: \{0,1\}^{n-k} \rightarrow \{0,1\}$ will be gained,
\begin{eqnarray}\label{subfunctions_fp}
f_p(w)&=&f(\underbrace{w_0w_1\cdots w_{n-k-1}}_{n-k} \underbrace{y_{p_{0}} y_{p_{1}}\cdots y_{p_{k-1}}}_{k}),
\end{eqnarray}
where $w=w_0w_1\cdots w_{n-k-1}\in\{0,1\}^{n-k}$ and $y_{p_{0}} y_{p_{1}}\cdots y_{p_{k-1}}\in\{0,1\}^{k}$ is the binary representation of $p\in\{0,1,\cdots,2^k-1\}$.

Prior to presenting the DEGGA, it is imperative that we disclose our partition strategy for the original function.

To begin with, it is essential to clarify that the original function will be partitioned into $t$ distinct segments, and the number of qubits in the $j$-th part is $n_j$, satisfying $\sum_{j=0}^{t-1} n_j =n$, where $j \in \{0,1,\cdots,t-1 \}$ and $2 \leq t \leq n$.

For the zeroth computing node ($j=0$), we choose last $\sum_{i=1}^{t-1} n_i$ bits in $x\in\{0,1\}^{n}$ to divide $f$, then we can get $2^{\sum_{i=1}^{t-1} n_i}$ subfunctions $f_{j,k}: \{0,1\}^{n_j} \rightarrow \{0,1\}$:
\begin{eqnarray}\label{fjk}
f_{j,k}(m_k)=f( \underbrace{m_k}_{n_j} \underbrace{y_{k_{0}} y_{k_{1}}\cdots y_{k_{\left( \sum_{i=1}^{t-1} n_i \right)-1}}}_{\sum_{i=1}^{t-1} n_i}),
\end{eqnarray}
where $m_k\in\{0,1\}^{n_j}$ and $y_{k_{0}} y_{k_{1}}\cdots y_{k_{\left( \sum_{i=1}^{t-1} n_i \right)-1}} \in\{0,1\}^{\sum_{i=1}^{t-1} n_i}$ is the binary representation of 
\begin{eqnarray}
k\in \left \{0,1,\cdots, 2^{\sum_{i=1}^{t-1} n_i}-1 \right\}.
\end{eqnarray} 
Subsequently, we construct a novel function $g_j: \{0,1\}^{n_j} \rightarrow \{0,1\}$ by utilizing these subfunctions,
\begin{eqnarray}\label{gj}
g_{j}(m_j)=\text{OR} \left(f_{j,0}(m_j), f_{j,1}(m_j), \cdots, f_{j, 2^{\sum_{i=1}^{t-1} n_i}-1}(m_j) \right),
\end{eqnarray}
where $m_j\in\{0,1\}^{n_j}$. In addition,
\begin{eqnarray}\label{orfunction}
\text{OR} (x) =
\begin{cases}
1,\vert x\vert \ge 1,\\
0,\vert x\vert =0,
\end{cases}
\end{eqnarray}
where $\vert x\vert $ is the Hamming weight of $x\in\{0,1\}^{n}$ (its number of 1's). 

For the $j$-th computing node ($j\in\{1,2,\cdots, t-2\}$), we choose first $\sum_{i=0}^{j-1} n_i$ and last $\sum_{i=j+1}^{t-1} n_i$ bits in $x\in\{0,1\}^{n}$ to divide $f$, then we can get $2^{\sum_{i=0}^{j-1} n_i+\sum_{i=j+1}^{t-1} n_i}$ subfunctions $f_{j,k}: \{0,1\}^{n_{j}} \rightarrow \{0,1\}$:
\begin{eqnarray}\label{fjk2}
f_{j,k}(m_k)=f( 
\underbrace{x_{k_{0}} x_{k_{1}}\cdots x_{k_{\left( \sum_{i=0}^{j-1} n_i \right)-1}}}_{\sum_{i=0}^{j-1} n_i} 
\underbrace{m_k}_{n_{j}} 
\underbrace{y_{k_{0}} y_{k_{1}}\cdots y_{k_{\left( \sum_{i=j+1}^{t-1} n_i \right)-1}}}_{\sum_{i=j+1}^{t-1} n_i}),
\end{eqnarray}
where $m_k\in\{0,1\}^{n_{j}}$ and $x_{k_{0}} x_{k_{1}}\cdots x_{k_{\left( \sum_{i=0}^{j-1} n_i \right)-1}}y_{k_{0}} y_{k_{1}}\cdots y_{k_{\left( \sum_{i=j+1}^{t-1} n_i \right)-1}} \in\{0,1\}^{\sum_{i=0}^{j-1} n_i+\sum_{i=j+1}^{t-1} n_i}$ is the binary representation of 
\begin{eqnarray}
k\in \left \{0,1,\cdots, 2^{\sum_{i=0}^{j-1} n_i+\sum_{i=j+1}^{t-1} n_i}-1 \right\}.
\end{eqnarray} 
Subsequently, we construct a novel function $g_j: \{0,1\}^{n_j} \rightarrow \{0,1\}$ by utilizing these subfunctions,
\begin{eqnarray}\label{gj2}
g_{j}(m_j)=\text{OR} \left(f_{j,0}(m_j), f_{j,1}(m_j), \cdots, f_{j, 2^{\sum_{i=0}^{j-1} n_i+\sum_{i=j+1}^{t-1} n_i}-1}(m_j) \right),
\end{eqnarray}
where $m_j\in\{0,1\}^{n_j}$. 

For the last computing node ($j=t-1$), we choose first $\sum_{i=0}^{t-2} n_i$ bits in $x\in\{0,1\}^{n}$ to divide $f$, then we can get $2^{\sum_{i=0}^{t-2} n_i}$ subfunctions $f_{j,k}: \{0,1\}^{n_j} \rightarrow \{0,1\}$:
\begin{eqnarray}\label{fjk3}
f_{j,k}(m_k)=f( 
\underbrace{y_{k_{0}} y_{k_{1}}\cdots y_{k_{\left( \sum_{i=0}^{t-2} n_i \right)-1}}}_{\sum_{i=0}^{t-2} n_i}
 \underbrace{m_k}_{n_j}),
\end{eqnarray}
where $m_k\in\{0,1\}^{n_j}$ and $y_{k_{0}} y_{k_{1}}\cdots y_{k_{\left( \sum_{i=0}^{t-2} n_i \right)-1}} \in\{0,1\}^{\sum_{i=0}^{t-2} n_i}$ is the binary representation of 
\begin{eqnarray}
k\in \left \{0,1,\cdots, 2^{\sum_{i=0}^{t-2} n_i}-1 \right\}.
\end{eqnarray} 
Subsequently, we construct a novel function $g_j: \{0,1\}^{n_j} \rightarrow \{0,1\}$ by utilizing these subfunctions,
\begin{eqnarray}\label{gj3}
g_{j}(m_j)=\text{OR} \left(f_{j,0}(m_j), f_{j,1}(m_j), \cdots, f_{j, 2^{\sum_{i=0}^{t-2} n_i}-1}(m_j) \right),
\end{eqnarray}
where $m_j\in\{0,1\}^{n_j}$. 

To summarize, for each node, we can derive the corresponding subfunctions $g_{j}(m_j)$, resulting in a total of $t$ subfunctions, where $j\in\{0,1,\cdots, t-1\}$ and $2 \leq t \leq n$. 

\subsection{Distributed Exact Generalized Grover's Algorithm with $2\leq t \leq n$ computing nodes}\label{degga}
In our study, we posit two fundamental assumptions. 

Firstly, we postulate that it is effortless to ascertain the number of targets to each subfunction $g_j (x)$, where $j\in\{0,1,\cdots, t-1\}$ and $2 \leq t \leq n$. However, this assumption does not equate to the process of determining the number of targets for the subfunctions as represented by Eq. \eqref{fjk}, Eq. \eqref{fjk2} and Eq. \eqref{fjk3}. This discrepancy arises from the fact that this assumption does not inherently reveal the underlying target strings.

Specifically, if we can acquire the set of target strings $A = \{x\in\{0,1\}^n | f(x)=1\}$ for an $n$-qubit generalized search problem, then the set of target strings for each subfunction $g_j (x)$ can be easily obtained by the partition strategy of DEGGA in subsection \ref{subfunctions}, and thus the number of target strings in that set can be determined.

Disregarding the aforementioned assumption, the recognized method for ascertaining the number of target strings for a Boolean function using quantum algorithms is the quantum counting algorithm \cite{RefBrassard2002}. Despite its widespread recognition, this algorithm is probabilistic rather than deterministic, and as such, we will not delve into this matter in great depth. The quest to accurately determine the number of objectives of a Boolean function via a deterministic quantum algorithm remains an unresolved inquiry.  

Let the number of objective terms of the $j$-th subfunction $g_j (x)$ be $a_j$, where
\begin{eqnarray}\label{ajitems}
a_j = \left|  A_{n_j} \right|
\end{eqnarray}
and
\begin{eqnarray}\label{Ajitems}
A_{n_j}=\{x\in\{0,1\}^{n_j}|g_j (x)=1\}.
\end{eqnarray}
Secondly, we assume that it is facile to procure the Oracle for each $g_j (x)$, which means we will have $t$ Oracles
\begin{eqnarray}\label{rtauj}
R_{g_j (x)}:\vert x\rangle\rightarrow e^{i\phi_{n_j} \cdot g_j (x)}\vert x\rangle,
\end{eqnarray}
where $x\in\{0,1\}^{n_j}$, $j\in\{0,1,\cdots, t-1\}$, $\phi_{n_j}=2\arcsin\left(\sin\left(\frac{\pi}{4J_{n_j}+6}\right) / \sin \left(\theta_{n_j}\right)\right)$, $J_{n_j}=\lfloor(\pi/2-\theta_{n_j})/(2\theta_{n_j})\rfloor$, $\theta_{n_j}= \arcsin \left({\sqrt{a_j /2^{n_j}}}\right)$.

Next, we offer a comprehensive exposition of the DEGGA with $2\leq t \leq n$ computing nodes in Algorithm \ref{algodegga}. The entire quantum circuit is depicted in Figure \ref{deggacir}.

\begin{breakablealgorithm}
\caption{Distributed Exact Generalized Grover's Algorithm with $2\leq t \leq n$ computing nodes}
\label{algodegga}
\begin{algorithmic}
\noindent \textbf{Input}: (1) The number of qubits $n$; (2) $N=2^n$; (3) The number of computing nodes $2\leq t \leq n$. (4) The number of qubits of the $j$-th computing node $n_{j}$, where $\sum_{j=0}^{t-1} n_j =n$ and $j \in \{0,1,\cdots,t-1 \}$; (5) A function $f(x)$, where $f(x)=0$ for $x\in\{0,1\}^n$ except $x\in A$, where $A = \{x\in\{0,1\}^n | f(x)=1\}$; (6) The number of target strings of $f(x)$ is $a=\left|  A \right|\ge 1$, which is a smaller number; (7) $t$ subfunctions $g_j (x)$ as in Eq.~\eqref{gj}, Eq.~\eqref{gj2} and Eq.~\eqref{gj3} generated according to $f(x)$ and $n$, where $j\in\{0,1,\cdots,t-1\}$; (8) The number of objective terms of the subfunction $g_j (x)$, $a_j$ as in Eq.~\eqref{ajitems} and Eq.~\eqref{Ajitems}; (9) The Oracle $R_{g_j (x)}$ corresponding to each subfunction $g_j (x)$, as in Eq.~\eqref{rtauj}.

\noindent\textbf{Output}: The target string $x\in A$ with a probability of $100 \%$.

\noindent \textbf{Procedure:}

\noindent \textbf{Step 1.} For $j$ from $0$ to $t-1$, initialize $n_j$ qubits $\vert 0\rangle^{\otimes n_j}$ for the $j$-th computing node. %and get 
%\begin{eqnarray}
%|\psi_0\rangle = \bigotimes_{j=0}^{t-1} |0\rangle^{\otimes n_j} = |0\rangle^{\otimes n}.
%\end{eqnarray} 

\noindent \textbf{Step 2.} For $j$ from $0$ to $t-1$, apply $n_j$ Hadamard gates $H^{\otimes n_j}$ on the $j$-th computing node. %We get 
%\begin{eqnarray}
%|\psi_1\rangle =\left(\bigotimes_{j=0}^{t-1} H^{\otimes n_j}\right)  |\psi_0\rangle= H^{\otimes n} |0\rangle^{\otimes n} =\frac{1}{\sqrt{N}}\sum\limits_{x\in\{0 , 1\}^n}|x\rangle.
%\end{eqnarray}

\noindent \textbf{Step 3.} For $j$ from $0$ to $t-1$, execute $L_{n_j}$ operator $J_{n_j}+1$ times on the $j$-th computing node, %and get 
%\begin{eqnarray}
%|\psi_1\rangle = \left(\bigotimes_{j=0}^{t-1} L_{n_j}^{J_{n_j}+1}\right)  |\psi_1\rangle,
%\end{eqnarray}
where
\begin{eqnarray}
L_{n_j}&=& -H^{\otimes n_j} R_{n_j} H^{\otimes n_j}R_{g_j (x)},\\
R_{n_j} &=&  I^{\otimes n_j} + \left(e^{i\phi_{n_j}}-1\right)(\vert 0\rangle\langle0\vert )^{\otimes n_j},\\
R_{g_j (x)} &=& I^{\otimes n_j} + \left(e^{i\phi_{n_j}}-1\right)\vert A_{n_j} \rangle\langle A_{n_j} \vert,\\
|A_{n_j}\rangle&=&\frac{1}{\sqrt{a_j}}\sum\limits_{x\in A_{n_j}}|x\rangle,\\
A_{n_j}&=&\{x\in\{0 , 1\}^{n_j} | g_j (x)=1\},\\
\phi_{n_j}&=&2\arcsin\left( \frac{\sin\left(\frac{\pi}{4J_{n_j}+6}\right)}{ \sin \left(\theta_{n_j}\right)} \right),\\
J_{n_j}&=&\left\lfloor \frac{\pi/2-\theta_{n_j}}{2\theta_{n_j}}\right \rfloor,\\
\theta_{n_j}&=&\arcsin \left({\sqrt{ \frac{a_j} {2^{n_j}}}}\right).
\end{eqnarray}

In fact, Step 2 and Step 3 can be viewed as a modified Grover's algorithm with $n_j$ qubits executed on the $j$-th computing node, denoted by
\begin{eqnarray}\label{ULoperator}
U_L &=& \left(\bigotimes_{j=0}^{t-1} L_{n_j}^{ J_{n_j}+1}\right) \left(\bigotimes_{j=0}^{t-1} H^{\otimes n_j}\right) \\ &=&  \bigotimes_{j=0}^{t-1} \left( L_{n_j}^{J_{n_j}+1}H^{\otimes n_j} \right) .
\end{eqnarray}

\noindent \textbf{Step 4.} Execute $R_{f'}$ operator, where
\begin{eqnarray}
R_{f'} = I^{\otimes n} + \left(e^{i\phi_{f'}}-1\right)\vert A \rangle\langle A \vert,
\end{eqnarray}
and
\begin{eqnarray}
\phi_{f'}&=&2\arcsin\left( \frac{\sin\left(\frac{\pi}{4J_{f'}+6}\right)}{ \sin \left(\theta_{f'}\right)} \right),\\ 
J_{f'}&=&\left \lfloor \frac{\pi/2-\theta_{f'}}{2\theta_{f'}} \right\rfloor, \\
 \theta_{f'}&=& \arcsin \left( \sqrt{\frac{a}{\prod_{j=0}^{t-1} a_j }} \right).
\end{eqnarray}

\noindent \textbf{Step 5.} For $j$ from $0$ to $t-1$, execute $L_{n_j}^{\dagger}$ operator $J_{n_j}+1$ times on the $j$-th computing node, where
\begin{eqnarray}
L_{n_j}^{\dagger}&=& -R_{g_j (x)}^{\dagger}H^{\otimes n_j} R_{n_j}^{\dagger} H^{\otimes n_j},\\
R_{n_j}^{\dagger} &=&  I^{\otimes n_j} + \left(e^{-i\phi_{n_j}}-1\right)(\vert 0\rangle\langle0\vert )^{\otimes n_j},\\
R_{g_j (x)} ^{\dagger}&=& I^{\otimes n_j} + \left(e^{-i\phi_{n_j}}-1\right)\vert A_{n_j} \rangle\langle A_{n_j} \vert.
\end{eqnarray}

\noindent \textbf{Step 6.} Repeat Step 2.

In fact, Step 5 and Step 6 can be regarded as conjugate transpose operations of $U_L$, represented as
\begin{eqnarray}
U^{\dagger}_L &=& \left[\left(\bigotimes_{j=0}^{t-1} L_{n_j}^{J_{n_j}+1}\right) \left(\bigotimes_{j=0}^{t-1} H^{\otimes n_j}\right) \right]^{\dagger}\\
&=& \left(\bigotimes_{j=0}^{t-1} H^{\otimes n_j}\right) \left(\bigotimes_{j=0}^{t-1} \left(L^{\dagger}_{n_j}\right)^{J_{n_j}+1}\right)\\
&=&  \bigotimes_{j=0}^{t-1} \left[ H^{\otimes n_j} \left(L^{\dagger}_{n_j}\right)^{J_{n_j}+1}\right] .
\end{eqnarray}

\noindent \textbf{Step 7.} Execute $R_{0'}$ operator, where
\begin{eqnarray}
R_{0'} =  I^{\otimes n} + \left(e^{i\phi_{f'}}-1\right)(\vert 0\rangle\langle0\vert )^{\otimes n}.
\end{eqnarray}

\noindent \textbf{Step 8.} Repeat Step 2 to Step 3. In other words, execute $U_L$ as in Eq.~\eqref{ULoperator} again.

\noindent \textbf{Step 9.} Repeat Step 4 to Step 8 $J_{f'}$ times.% and get 
%\begin{eqnarray}
%|\psi_3\rangle = \frac{1}{\sqrt{a}}\sum\limits_{x\in A}|x\rangle = |A\rangle.
%\end{eqnarray}

\noindent \textbf{Step 10.} Measure each qubit in the basis $\{\vert 0\rangle, \vert 1\rangle\}$.
\end{algorithmic}
\end{breakablealgorithm}

\begin{figure}[t]
    \centering
\setlength{\abovecaptionskip}{-0.6cm}
\includegraphics[width=1\textwidth]{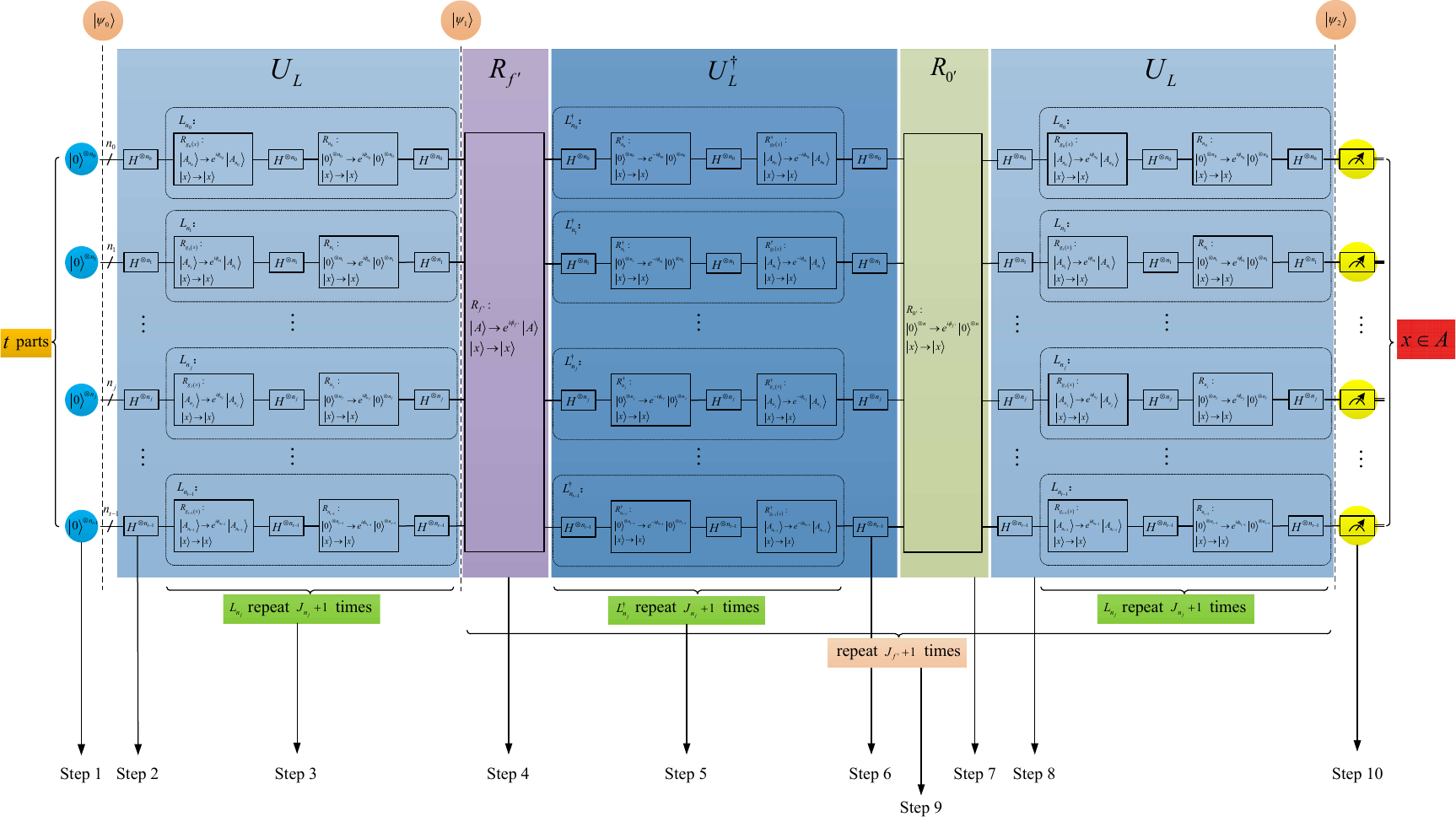}
\caption{Quantum circuit of the distributed exact generalized Grover's algorithm with $2\leq t \leq n$ computing nodes.}
\label{deggacir}
\end{figure}

It is noteworthy that in some cases, $R_{f'}$ can choose to execute quantum gates between computational nodes instead of $n$-qubit gates, which can improve the landability of DEGGA. Eventually, the target strings $x\in A$ can be successfully searched exactly by following the aforementioned algorithmic steps. Consequently, the entire process of the DEGGA with $2\leq t \leq n$ computing nodes is fulfilled.

%It is noteworthy that $R_{f'}$ can be selectively utilized to implement quantum gates between computational nodes as deemed fit. 

\section{Analysis}\label{analysis}
In this section, we will demonstrate the correctness of the DEGGA. Afterwards, to illustrate the advantageous effects of distributed quantum algorithms, the circuit depth of our suggested approach will be rigorously offered. In conclusion, we will provide a comparison between our scheme with the modified Grover's algorithm and the existing distributed Grover's algorithms.

\subsection{Correctness}\label{correctness}
To validate the correctness of DEGGA, it is imperative to demonstrate that the target string $x\in A$ can be derived through Algorithm \ref{algodegga} with a theoretical success probability of $100\%$, signifying its precise attainment. The correctness of the DEGGA is substantiated by Theorem  \ref{proofcorrectness} presented subsequently.

\begin{theorem}\label{proofcorrectness}
(\textbf{Correctness of DEGGA}) For an $n$-qubit generalized search problem, it is associated with the following Boolean function:
\begin{eqnarray}
f(x)=
\begin{cases}
1,x\in A,\\
0,x\notin A,
\end{cases}
\end{eqnarray}
where $A = \{x\in\{0,1\}^n | f(x)=1\}$. The target string $x\in A$ can be exactly obtained by applying Algorithm \ref{algodegga}.
\end{theorem}

\begin{proof}
First, let the number of target strings for the Boolean function $f: \{0,1\}^n \rightarrow \{0,1\}$ corresponding to the $n$-qubit  generalized search problem being
\begin{eqnarray}
a = |A| \ge 1,
\end{eqnarray}
where $A = \{x\in\{0,1\}^n | f(x)=1\}$.
We suppose that the $a$ targets are 
\begin{eqnarray}
x^{(k)} = x _{0}^{(k)}x_{1}^{(k)}\cdots x_{n-1}^{(k)}\in A,
\end{eqnarray}
respectively, where $x_{i}^{(k)}\in\{0,1\}$, $i\in\{0,1,\cdots,n-1\}$ and $k\in\{0,1,\cdots,a-1\}$.
Since the DEGGA has a total of $t$ computing nodes, let the number of qubits of the $j$-th computing node being $n_{j}$, where $\sum_{j=0}^{t-1} n_j =n$ and $2\leq t \leq n$.

Second, in accordance with the partition strategy of DEGGA in subsection \ref{subfunctions}, a total of $t$ subfunctions $g_j(x)$ are generated as derived from $f(x)$ and $n$, as delineated in Eq.~\eqref{gj}, Eq.~\eqref{gj2} and Eq.~\eqref{gj3}, where $j\in\{0,1,\cdots,t-1\}$. For the subfunction $g_j(x)$, its target string is in the following set
\begin{eqnarray}
\left\{x _{\sum_{i=0}^{j-1}n_i}^{(k)} x _{\left(\sum_{i=0}^{j-1}n_i\right)+1}^{(k)} x _{\left(\sum_{i=0}^{j-1}n_i\right)+2}^{(k)} \cdots x _{\left(\sum_{i=0}^{j}n_i\right)-1}^{(k)} \right\} \subseteq \{0,1\}^{n_j},
\end{eqnarray}
respectively, where $k\in\{0,1,\cdots,a-1\}$. Let the number of objective terms of the subfunction $g_j (x)$ being $a_j$ as in Eq.~\eqref{ajitems} and Eq.~\eqref{Ajitems}.

Third, as deduced from subsection \ref{sectionlong}, it is evident that the modified Grover's algorithm can precisely ascertain the target strings of the $n$-qubit Boolean function $f(x)$. In other words, once we know each Boolean subfunction and its number of target strings, then we can accurately obtain the quantum superposition state corresponding to the targets of that Boolean subfunction by means of the modified Grover's algorithm on a small scale.

Fourth, from Step 1 to Step 3 of the DEGGA, it is clear that a small scale modified Grover's algorithm is executed, which means that each node can obtain exactly the quantum superposition state of the targets of its corresponding Boolean subfunction $g_j(x)$. It means we will obtain
\begin{eqnarray}
|\psi_1\rangle &=&  U_L |\psi_0\rangle \\
&=& \left[ \bigotimes_{j=0}^{t-1} \left( L_{n_j}^{J_{n_j}+1}H^{\otimes n_j} \right)  \right]  |0\rangle^{\otimes n} \\
&=& \bigotimes_{j=0}^{t-1} \left[ \left( L_{n_j}^{J_{n_j}+1}H^{\otimes n_j} \right)  |0\rangle^{\otimes n_j} \right]   \\
&=& \bigotimes_{j=0}^{t-1} \left (\frac{1}{\sqrt{a_j}} \sum_{k=0}^{a_j-1} \underbrace{ \Ket{ x _{\sum_{i=0}^{j-1}n_i}^{(k)} x _{\left(\sum_{i=0}^{j-1}n_i\right)+1}^{(k)} \cdots x _{\left(\sum_{i=0}^{j}n_i\right)-1}^{(k)} } }_{n_j} \right)
\end{eqnarray}
after Step 3 of the DEGGA.

Finally, Step 4 to Step 9 of the DEGGA can be regarded as analogous to those of the modified Grover's algorithm (or as a process of exact amplitude amplification). Crucially, this necessitates the clarification of the quantity of quantum superposition states and the number of target strings, which allows to obtain the rotation parameter $\phi_{f'}$. At this juncture, the number of quantum superposition states is given by $\prod_{j=0}^{t-1} a_j $ rather than $2^n$, whereas the number of target strings remains consistent with the initial count, denoted by $a$. That is to say, after Step 9 of the DEGGA, we will acquire
%Finally, Step 4 to Step 9 of the DEGGA are the exact amplitude amplification process for the superposition of $\prod_{j=0}^{t-1} a_j $ quantum states rather than $2^n$ quantum states. That is to say, we will acquire
\begin{eqnarray}
|\psi_2\rangle&=& \left(U_L R_{0'} U^{\dagger}_L R_{f'}\right) ^{J_{f'}+1} |\psi_1\rangle\\
&=&\frac{1}{\sqrt{a}} \sum_{k=0}^{a-1}  \Ket{x _{0}^{(k)}x_{1}^{(k)}\cdots x_{n-1}^{(k)} } \\
&=&\frac{1}{\sqrt{a}}\sum\limits_{x\in A}|x\rangle = |A\rangle,
\end{eqnarray}
where 
\begin{eqnarray}
J_{f'}&=&\left \lfloor \frac{\pi/2-\theta_{f'}}{2\theta_{f'}} \right\rfloor, \\
 \theta_{f'}&=& \arcsin \left( \sqrt{\frac{a}{\prod_{j=0}^{t-1} a_j }} \right)\\
\phi_{f'}&=&2\arcsin\left( \frac{\sin\left(\frac{\pi}{4J_{f'}+6}\right)}{ \sin \left(\theta_{f'}\right)} \right).
\end{eqnarray}

Thus, by measuring each qubit of $|\psi_2\rangle$ in the basis $\{\vert 0\rangle, \vert 1\rangle\}$ (Step 10 of the DEGGA), we are able to retrieve the target strings corresponding to the original Boolean function with a probability of $100\%$, which is referred to as an exact outcome.
\end{proof}

\subsection{Circuit depth}\label{circuitdepth}
In order to underscore the advantages of distributed quantum algorithms, we initially furnish the definition of circuit depth, followed by a rigorous specification of the circuit depth pertaining to the DEGGA.

\begin{definition} \label{defdepth}
(\textbf{Depth of circuit}) The depth of a circuit is characterized by the longest sequential path from input to output, advancing temporally along the qubit wires.
\end{definition}

For instance, the depth of the circuit on the left is 1, whereas that on the right is 11. (see Figure \ref{depthexample})

\begin{figure}[H]
\centering
\includegraphics[width=4in]{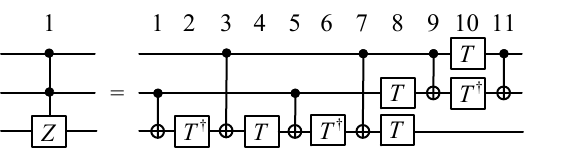}
\caption{\label{depthexample} Equivalent quantum circuit representation of the $C^2 Z$ gate, where $T$ denotes the $\pi / 8$ gate, and $T^{\dagger}$ represents the conjugate transpose of $T$.}
\end{figure}

In \cite{RefZhou2023DEGA}, they have rigorously demonstrated the corresponding circuit depths for both the modified Grover's algorithm and the distributed Grover's algorithm (DEGA) in scenarios involving a single target, as delineated in Theorem \ref{LongDEGAdepth}.

\begin{theorem}
For single target scenarios, the circuit depths for the modified Grover's algorithm and the distributed Grover's algorithm (DEGA) are
\begin{eqnarray}
dep(\text{modified Grover, single target}) = 9 + 8\left\lfloor \frac{\pi}{4}\sqrt{2^n} - \frac{1}{2}\right\rfloor,
\end{eqnarray} 
and
\begin{eqnarray}
dep(\text{DEGA, single target}) = 8(n~\text{mod}~2)+9 =
\begin{cases}
9,          &n~\text{is even},\\
17,       &n~\text{is odd},
\end{cases}
\end{eqnarray}
respectively.
\label{LongDEGAdepth}
\end{theorem}

\begin{proof}
For brevity, the detailed proof is omitted and can be found in  \cite{RefZhou2023DEGA}.
\end{proof}

For scenarios involving multiple targets, we can employ analogous methodologies to determine the circuit depth of the modified Grover's algorithm (see Theorem \ref{Longmultipledepth}).

\begin{theorem}
For multiple targets scenarios, the circuit depth for the modified Grover's algorithm is
\begin{eqnarray}
dep(\text{modified Grover, multiple targets}) = (6+3a) +  (5+3a) \cdot \left\lfloor \frac{\pi}{4}\sqrt{\frac{2^n}{a}} - \frac{1}{2}\right\rfloor,
\end{eqnarray} 
where $a$ represents the number of target strings for the original $n$-qubit Boolean function $f: \{0,1\}^n \rightarrow \{0,1\}$.
\label{Longmultipledepth}
\end{theorem}

\begin{proof}
First, let the number of target strings for the original $n$-qubit Boolean function $f: \{0,1\}^n \rightarrow \{0,1\}$ being
\begin{eqnarray}
a = |A| \ge 1,
\end{eqnarray}
where $A = \{x\in\{0,1\}^n | f(x)=1\}$.

Second, operator $L = -H^{\otimes n} R_0 H^{\otimes n}R_{f(x)}$ is executed with $J+1$ times, where 
\begin{eqnarray}
J&=&\left\lfloor \frac{\pi/2-\theta}{2\theta}\right \rfloor=\left\lfloor \frac{\pi}{4\theta}- \frac{1}{2}\right \rfloor,\\
\theta&=&\arcsin \left({\sqrt{ \frac{a} {2^{n}}}}\right) \sim {\sqrt{ \frac{a} {2^{n}}}} ~ (n \rightarrow + \infty ).
\end{eqnarray}
It means
\begin{eqnarray}
J=\left\lfloor \frac{\pi}{4\theta}- \frac{1}{2}\right \rfloor=\left\lfloor \frac{\pi}{4}\sqrt{\frac{2^n}{a}} - \frac{1}{2}\right\rfloor.
\end{eqnarray}

Third, it is easy to obtain
\begin{eqnarray}
 dep(H^{\otimes n})&=&1,\\
 dep(R_f)&=&3a, \\
 dep(R_0)&=&3.
\end{eqnarray}

Finally, the circuit depth of the modified Grover's algorithm for multiple targets scenarios is
\begin{eqnarray}
dep(\text{modified Grover, multiple targets}) &=& 1 + \left( \left\lfloor \frac{\pi}{4}\sqrt{\frac{2^n}{a}} - \frac{1}{2}\right\rfloor +1 \right) \cdot (3a+1+3+1) \\
                 &=& (6+3a) +  (5+3a) \cdot \left\lfloor \frac{\pi}{4}\sqrt{\frac{2^n}{a}} - \frac{1}{2}\right\rfloor.
\end{eqnarray}
\end{proof}
It can be seen that if $a$ is fixed, the circuit depth of the modified Grover's algorithm still deepens as $n$ increases. Besides, it should be noted that in the calculation of $dep(R_f)$ mentioned above, the case of all zero string being the target state and the process of circuit optimization were ignored.

Next, we present the circuit depth of the DEGGA as articulated by the following theorem.

\begin{theorem}
(\textbf{The circuit depth for DEGGA}) The circuit depth for Algorithm \ref{algodegga} is 
\begin{eqnarray}
&&dep(\text{DEGGA, multiple targets})\nonumber\\
&=& \left(2J_{f'}+3 \right) \cdot \max_{j \in \{0,1,\cdots,t-1 \}} \left\{ dep(\text{modified Grover, $a_j$ targets})\right\} + \left(J_{f'}+1 \right)\cdot (3a+3),
\end{eqnarray}
where
\begin{eqnarray}
J_{f'}&=&\left \lfloor \frac{\pi/2-\theta_{f'}}{2\theta_{f'}} \right\rfloor, \\
 \theta_{f'}&=& \arcsin \left( \sqrt{\frac{a}{\prod_{j=0}^{t-1} a_j }} \right),
\end{eqnarray}
and
\begin{eqnarray}
&&\max_{j \in \{0,1,\cdots,t-1 \}} \left\{ dep(\text{modified Grover, $a_j$ targets})\right\} \nonumber \\
&=&  \max_{j \in \{0,1,\cdots,t-1 \}} \left\{ (6+3a_j) +  (5+3a_j) \cdot \left\lfloor \frac{\pi}{4}\sqrt{\frac{2^{n_j}}{a_j}} - \frac{1}{2}\right\rfloor  \right\}.
\end{eqnarray}
In addition, $a$ represents the number of target strings for the original $n$-qubit Boolean function $f: \{0,1\}^n \rightarrow \{0,1\}$, $a_j$ represents the number of objective terms of the $j$-th subfunction $g_j (x)$ (as in Eq.~\eqref{gj}, Eq.~\eqref{gj2} and Eq.~\eqref{gj3}) generated according to $f(x)$ and $n$, where $j\in\{0,1,\cdots,t-1\}$.
\label{DEGGAdepth}
\end{theorem}

\begin{proof}
First, let the number of target strings for the original $n$-qubit Boolean function $f: \{0,1\}^n \rightarrow \{0,1\}$ being
\begin{eqnarray}
a = |A| \ge 1,
\end{eqnarray}
where $A = \{x\in\{0,1\}^n | f(x)=1\}$. Let the number of objective terms of the $j$-th subfunction $g_j (x)$ (as in Eq.~\eqref{gj}, Eq.~\eqref{gj2} and Eq.~\eqref{gj3} generated according to $f(x)$ and $n$ being
\begin{eqnarray}
a_j = \left|  A_{n_j} \right| ,
\end{eqnarray}
where $\sum_{j=0}^{t-1} n_j =n$, $A_{n_j}=\{x\in\{0,1\}^{n_j}|g_j (x)=1\}$, and $j \in \{0,1,\cdots,t-1 \}$.

Second, the circuit depth for DEGGA clearly satisfies 
\begin{eqnarray}
dep(\text{DEGGA, multiple targets}) =  \sum_{j=2}^{9} dep(\text{Step $j$, DEGGA}).
\end{eqnarray}

Third, given that each computational node will simultaneously execute the quantum circuit of the small-scale modified Grover's algorithm, it suffices to solely consider the computational node with the largest depth when determining the depth of the circuit. From the conclusion of Theorem \ref{Longmultipledepth}, we can easily obtain the circuit depth in from Step 2 to Step 3 for DEGGA,
\begin{eqnarray}
&&dep(\text{Step 2, DEGGA}) + dep(\text{Step 3, DEGGA}) \nonumber\\
&=&\max_{j \in \{0,1,\cdots,t-1 \}} \left\{ dep(\text{modified Grover, $a_j$ targets})\right\} \\
&=& \max_{j \in \{0,1,\cdots,t-1 \}} \left\{ (6+3a_j) +  (5+3a_j) \cdot \left\lfloor \frac{\pi}{4}\sqrt{\frac{2^{n_j}}{a_j}} - \frac{1}{2}\right\rfloor  \right\}.
\end{eqnarray}
Similarly, it can be inferred that
\begin{eqnarray}
&& dep(\text{Step 5, DEGGA}) + dep(\text{Step 6, DEGGA}) \nonumber\\
&=& dep(\text{Step 8, DEGGA})\\
&=&dep(\text{Step 2, DEGGA}) + dep(\text{Step 3, DEGGA}).
\end{eqnarray}
Besides, we can similarly obtain the circuit depths in Step 4 and Step 6 for DEGGA,
\begin{eqnarray}
dep(\text{Step 4, DEGGA}) &=& dep(R_f')=3a, \\
dep(\text{Step 7, DEGGA}) &=& dep(R_0')=3.
\end{eqnarray}

Fourth, since Step 9 (Step 4 to Step 8) needs to be repeated a total of $J_{f'}$ times, where
\begin{eqnarray}
J_{f'}&=&\left \lfloor \frac{\pi/2-\theta_{f'}}{2\theta_{f'}} \right\rfloor, \\
 \theta_{f'}&=& \arcsin \left( \sqrt{\frac{a}{\prod_{j=0}^{t-1} a_j }} \right),
\end{eqnarray}
its corresponds to a circuit depth of
\begin{eqnarray}
dep(\text{Step 9, DEGGA}) &=& J_{f'} \cdot \sum_{j=4}^{8} dep(\text{Step $j$, DEGGA})\\
&=& J_{f'} \cdot  \left(2 \cdot \max_{j \in \{0,1,\cdots,t-1 \}} \left\{ dep(\text{modified Grover, $a_j$ targets})\right\}+ 3a + 3 \right)
\end{eqnarray}

Finally, the circuit depth for Algorithm \ref{algodegga} is 
\begin{eqnarray}
&&dep(\text{DEGGA, multiple targets}) \nonumber\\
&=& \sum_{j=2}^{9} dep(\text{Step $j$, DEGGA})  \\
&=& \sum_{j=2}^{8} dep(\text{Step $j$, DEGGA}) +dep(\text{Step 9, DEGGA})\\
&=& \left( 3\cdot \max_{j \in \{0,1,\cdots,t-1 \}} \left\{ dep(\text{modified Grover, $a_j$ targets})\right\} +3a+3 \right) \nonumber \\
&& + J_{f'} \cdot  \left(2 \cdot \max_{j \in \{0,1,\cdots,t-1 \}} \left\{ dep(\text{modified Grover, $a_j$ targets})\right\}+ 3a + 3 \right)\\
&=& \left(2J_{f'}+3 \right) \cdot \max_{j \in \{0,1,\cdots,t-1 \}} \left\{ dep(\text{modified Grover, $a_j$ targets})\right\} + \left(J_{f'}+1 \right)\cdot (3a+3)\\
&=& \left(2J_{f'}+3 \right) \cdot  \max_{j \in \{0,1,\cdots,t-1 \}} \left\{ (6+3a_j) +  (5+3a_j) \cdot \left\lfloor \frac{\pi}{4}\sqrt{\frac{2^{n_j}}{a_j}} - \frac{1}{2}\right\rfloor  \right\}+ \left(J_{f'}+1 \right)\cdot (3a+3).
\end{eqnarray}
\end{proof}

So far, we have successfully determined the circuit depth of the DEGGA. It can be seen that with a fixed $a$, the choice of dividing strategy is crucial, as it directly determines $t$ (the number of nodes), $n_j$ (the number of qubits on the $j$-th node), and $a_j$ (the number of target strings on the $j$-th node), which are the key factors affecting the circuit depth of the DEGGA.

To put it another way, $n$ is not the most essential factor in determining the circuit depth of DEGGA. Instead, if there are more computing nodes, the qubits of each node are more evenly distributed, the circuit depth advantage of DEGGA can be better reflected.

Last but not least, in comparison to the modified Grover's algorithm, DEGGA substantially diminishes the utilization of $\text{C}^{n-1}\text{PS}$ gates. If we further consider the decomposition of $\text{C}^{n-1}\text{PS}$ gates into single-qubit gates and double-qubit gates, the disparity between the circuit depth requisite for the modified Grover's algorithm and those essential for DEGGA will be exacerbated. We will cover in subsection \ref{sectionDecomposition}. 

\subsection{Comparison}\label{comparison}
Within this subsection, we will present a comparative analysis between DEGGA and the existing distributed Grover's algorithms.

In \cite{RefAvron2021}, Avron et al. decomposed the original $n$-qubit Boolean function $f: \{0,1\}^n \rightarrow \{0,1\}$ into two subfunctions $f_{\text{even}/\text{odd}}: \{0,1\}^{n-1} \rightarrow \{0,1\}$ as in Eq.~\eqref{subfunctions_pra1} and Eq.~\eqref{subfunctions_pra2}. For each subfunction, the Grover's algorithm with $(n-1)$ qubits was executed separately. Subsequently, they will determine whether to assign 0 or 1 to the $i$-th bit, contingent upon which subfunction is capable of yielding the correct outcome (0 for $f_{\text{even}}$ and 1 for $f_{\text{odd}}$). In other words, they have merely obtained the $n-1$ qubits of the target state, as opposed to directly acquiring the target state itself.

In summary, the proposed methodology they advocate is compatible with a 2-node system and possesses the capability to address single-objective search problems (since the problem of how to determine the number of objectives for each subfunction is not addressed). Nonetheless, aside from particular instances, an exact search cannot be realized, and the aggregate number of qubits employed amounts to $2(n-1)$.

In \cite{RefQiu2022}, Qiu et al. further extended the procedure described above. They divided the initial Boolean function $f$ into $2^k$ subfunctions $f_p: \{0,1\}^{n-k} \rightarrow \{0,1\}$ as in Eq.~\eqref{subfunctions_fp}, where $2^k$ represents the number of computing nodes. Having ascertained the number of solutions within each subfunction via the quantum counting algorithm, they subsequently executed the $(n-k)$-qubit Grover's algorithm for each individual subfunction. This process can be conducted either serially or in parallel. Similarly, this approach can only acquire $n-k$ bits of the target string, and subsequently determines the remaining bits in the target string contingent on the specific subfunction involved.

To sum up, their suggested scheme is also equipped to tackle multi-objective search problems and has been further generalized for application in systems with $2^k$ nodes. Likewise, with the exception of special cases where an exact search is not feasible. The aggregate number of qubits necessary for the parallel mode is $2^k(n-k)$. Here we ignore the calculation of the number of qubits required to run the quantum counting algorithm.

In \cite{RefZhou2023DEGA}, Zhou and Qiu et al. adopted a method similar to that presented in this paper for the generation of subfunctions relevant to each node, creating a total of $\lfloor n/2 \rfloor$ subfunctions $g_j (m_j)$ based on the original function $f(x)$ and the value of $n$, where $j \in\{0,1,\cdots, \lfloor n/2 \rfloor-1\}$. For $j \in\{0,1,\cdots,$ $\lfloor n/2 \rfloor-2\}$, they implemented the 2-qubit Grover's algorithm. For $j \in\{ \lfloor n/2 \rfloor-1\}$, they executed the 2-qubit Grover's algorithm when $n$ was even, and employed the 3-qubit modified Grover's algorithm for odd values of $n$. Each component contributes a minor segment of the target, and the entire target string can be constructed by aggregating all the substrings. In essence, each node requires only two or three qubits, thereby facilitating a more straightforward physical implementation.

In short, the strategy they advocate is applicable to the system with $\lfloor n/2 \rfloor$ nodes. Moreover, they are capable of achieving precise searches for single target, yet fall short in multi-target situations. The employment of additional auxiliary qubits is unnecessary. Within their framework, the total qubits remains at $n$.

In this paper, the DEGGA we designed can accurately resolve generalized search problems that comprises multiple targets within an unordered database. Additionally, this method exhibits applicability to systems with any number of $2 \leq t \leq n$ nodes, thereby demonstrating its scalability. Last but not least, the total qubits count in our methodology is $n$, indicating that the employment of auxiliary qubits is superfluous.

To facilitate a more direct comparison between various distributed methods, we present the following table for elucidation (see Table \ref{algorithmscompare}).

\begin{table}[H]
  \centering
  \caption{A simple comparison between DEGGA and the existing distributed Grover's algorithms.}
\scalebox{0.7}{
  \begin{tabular}{lccccc}
    \toprule
   &The algorithm in \cite{RefAvron2021}& The algorithm in \cite{RefQiu2022}& The algorithm in  \cite{RefZhou2023DEGA}& \textbf{DEGGA}  \\
    \midrule
\text{1. The number of computing nodes.} &2&$2^k$, where $1 \leq k \leq n-1$& $\lfloor n/2 \rfloor$ &$\bm{2 \leq t \leq n}$ \\
\text{2. Does it solve the generalized search problem?} &\text{No} &\text{Yes}&\text{No}&\textbf{Yes}\\
\text{3. Does it solve exactly?}
&\text{No}&\text{No}&\text{Yes}&\textbf{Yes} \\
\text{4. The total number of qubits.}&$2(n-1)$&$2^k(n-k)$&$n$&$\bm{n}$\\
\text{5. Does it require auxiliary qubits?}& \text{No}&\text{Yes}&\text{No}&\textbf{No}\\
\text{6. Does it implement through quantum simulation software?}
&\text{Yes}&\text{No}&\text{Yes}&\textbf{Yes} \\
    \bottomrule
  \end{tabular}}
\label{algorithmscompare}
\end{table}

\section{Experiment}\label{experiment}
Within this section, we are going to start with a distinct generalized search problem that encompasses two target strings (000000 and 111111). Subsequently, we will illustrate the execution of 6-qubit modified Grover's algorithm. Moving forward, we shall sequentially demonstrate the detailed implementation process of 6-qubit DEGGA for both 2-node and 3-node configurations. The efficacy and practicality of our proposed method can be further corroborated by running the quantum circuits on the MindSpore Quantum. Through the decomposition of multi-qubit gates in quantum circuits, it becomes evident that distributed quantum algorithms possess exclusive advantages.

\subsection{The 6-qubit modified Grover's algorithm, target string $x \in \{000000,111111 \}$}\label{6qmga}

Let Boolean function $f: \{0,1\}^6 \rightarrow \{0,1\}$. Suppose
\begin{eqnarray}
f(x)=
\begin{cases}
1,x=000000,111111,\\
0, \text{otherwise},
\end{cases}
\end{eqnarray}
where $x\in\{0,1\}^6$ and $x \in \{000000,111111 \}$ is the target string. For $n=6$, it is feasible to construct the quantum circuit of 6-qubit modified Grover's algorithm. (see Figure \ref{6qubit-long-cir}). 

\begin{figure}[H]
\centering
\includegraphics[width=0.95\textwidth]{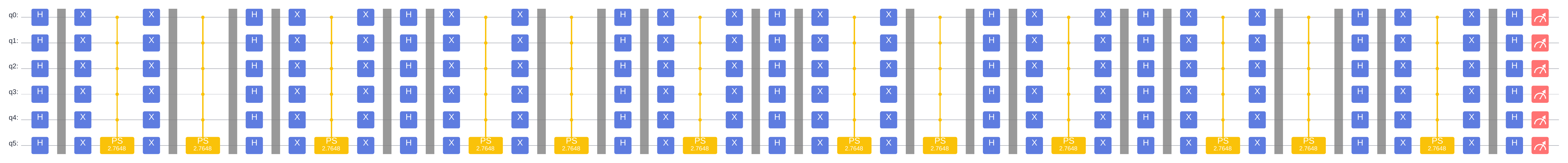}
\caption{Quantum circuit of 6-qubit modified Grover's algorithm (target string $x \in \{000000,111111 \}$). The parameter in the circuit is $\phi=2\arcsin\left(\sin\left(\frac{\pi}{4J+6}\right) / \sin \theta\right) =2.764763603060391\approx 2.7648$, where $J=\lfloor(\pi/2-\theta)/(2\theta)\rfloor$ and $\theta= \arcsin \left({\sqrt{1/2^6}}\right)$.}
\label{6qubit-long-cir}
\end{figure}

Please note that the PhaseShift, referred to as the PS gate, is a single-qubit gate with the following matrix representation:
\begin{eqnarray}
\text{PS($\phi$)} =
\left[
\begin{array}{cc}
1 & 0 \\
0 & e^{i\phi}
\end{array}
\right].
\end{eqnarray}
By sampling 10,000 times of the circuit depicted in Figure \ref{6qubit-long-cir}, the sampling outcomes are presented in Figure \ref{6qubit-long-out}. Throughout this paper, we select the backend simulator named ``mqvector''. To enhance the reproducibility of our experimental findings, we have established a random seed value of 42 for the sampling process.

\begin{figure}[H]
\centering
\includegraphics[width=0.6\textwidth]{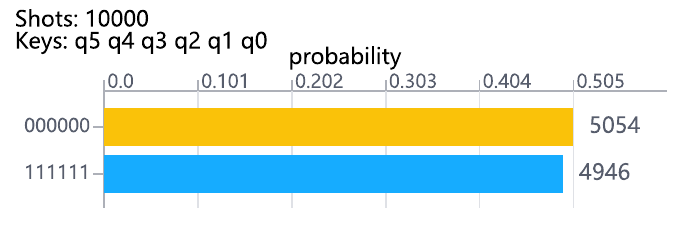}
\caption{The sampling results of 6-qubit modified Grover's algorithm (target string $x \in \{000000,111111 \}$). }
\label{6qubit-long-out}
\end{figure}

The sampling outcomes corroborate that the measurement precisely yields $x$ within the set $\{000000, 111111\}$. Furthermore, the entire quantum circuit consists of a total of 162 quantum gates, including 12 $\text{C}^5\text{PS}$ gates, and has a circuit depth of 37.

\subsection{The 6-qubit DEGGA  with $t = 2$ computing nodes, target string $x \in \{000000,111111 \}$}\label{6qdegga2}
Next, we will present a specific implementation of DEGGA utilizing two computing nodes. In this instance, we configure the computing nodes $t=2$, and assign an equal number of qubits to each computing node, setting $n_0 = n_1 =3$.

On one hand, we choose last $3$ bits in $x$ to divide $f$, then we can get $2^3=8$ subfunctions $f_{{0},j}: \{0,1\}^3 \rightarrow \{0,1\}$:
\begin{eqnarray}
f_{{0},0}(m_0)&=f(  m_0 000), f_{{0},4}(m_0)&=f(  m_0 100),\\
f_{{0},1}(m_0)&=f(  m_0 001), f_{{0},5}(m_0)&=f(  m_0 101),\\
f_{{0},2}(m_0)&=f(  m_0 010), f_{{0},6}(m_0)&=f(  m_0 110),\\
f_{{0},3}(m_0)&=f(  m_0 011), f_{{0},7}(m_0)&=f(  m_0 111),
\end{eqnarray}
where $m_0\in\{0,1\}^{3}$ and $j\in\{0,1, \cdots,7\}$. Obviously, $f_{{0},1}(m_0)=f_{{0},2}(m_0)=f_{{0},3}(m_0)=f_{{0},4}(m_0)=f_{{0},5}(m_0)=f_{{0},6}(m_0)\equiv0$, and
\begin{eqnarray}
f_{{0},0}(m_0)&=&
\begin{cases}
1,m_0=000,\\
0,m_0\neq000,
\end{cases}\\
f_{{0},7}(m_0)&=&
\begin{cases}
1,m_0=111,\\
0,m_0\neq111,
\end{cases}
\end{eqnarray}
where $m_0\in\{0,1\}^3$.

Afterwards, we generate a new function $g_0: \{0,1\}^3 \rightarrow \{0,1\}$ through above eight subfunctions,
\begin{eqnarray}
g_{0}(m_0)=\text{OR} \left(f_{{0},0}(m_0), f_{{0},1}(m_0), \cdots , f_{{0},7}(m_0)\right)=
\begin{cases}
1,m_0=000,111,\\
0, \text{otherwise},
\end{cases}
\end{eqnarray}
where $m_0\in\{0,1\}^{3}$.

On the other hand, we choose first $3$ bits in $x$ to divide $f$, then we can get $2^3=8$ subfunctions $f_{{1},j}: \{0,1\}^3 \rightarrow \{0,1\}$:
\begin{eqnarray}
f_{{1},0}(m_1)&=f( 000 m_1 ), f_{{1},4}(m_1)&=f( 100 m_1 ),\\
f_{{1},1}(m_1)&=f( 001 m_1 ), f_{{1},5}(m_1)&=f( 101 m_1 ),\\
f_{{1},2}(m_1)&=f( 010 m_1 ), f_{{1},6}(m_1)&=f( 110 m_1 ),\\
f_{{1},3}(m_1)&=f( 011 m_1 ), f_{{1},7}(m_1)&=f( 111 m_1 ),
\end{eqnarray}
where $m_1\in\{0,1\}^{3}$ and $j\in\{0,1, \cdots,7\}$. Obviously, $f_{{1},1}(m_1)=f_{{1},2}(m_1)=f_{{1},3}(m_1)=f_{{1},4}(m_1)=f_{{1},5}(m_1)=f_{{1},6}(m_1)\equiv0$, and
\begin{eqnarray}
f_{{1},0}(m_1)&=&
\begin{cases}
1,m_1=000,\\
0,m_1\neq000,
\end{cases}\\
f_{{1},7}(m_1)&=&
\begin{cases}
1,m_1=111,\\
0,m_1\neq111,
\end{cases}
\end{eqnarray}
where $m_1\in\{0,1\}^3$.

Afterwards, we generate another new function $g_1: \{0,1\}^3 \rightarrow \{0,1\}$ through above eight subfunctions,
\begin{eqnarray}
g_{1}(m_1)=\text{OR} \left(f_{{1},0}(m_1), f_{{1},1}(m_1), \cdots , f_{{1},7}(m_1)\right)=
\begin{cases}
1,m_1=000,111,\\
0, \text{otherwise},
\end{cases}
\end{eqnarray}
where $m_1\in\{0,1\}^{3}$.

Given our assumption that it is straightforward to acquire the Oracle for each subfunction, we can construct the complete quantum circuit of 6-qubit DEGGA with $t=2$ computing nodes (see Figure \ref{6qubit-degga2-cir}).

\begin{figure}[H]
\centering
\includegraphics[width=0.95\textwidth]{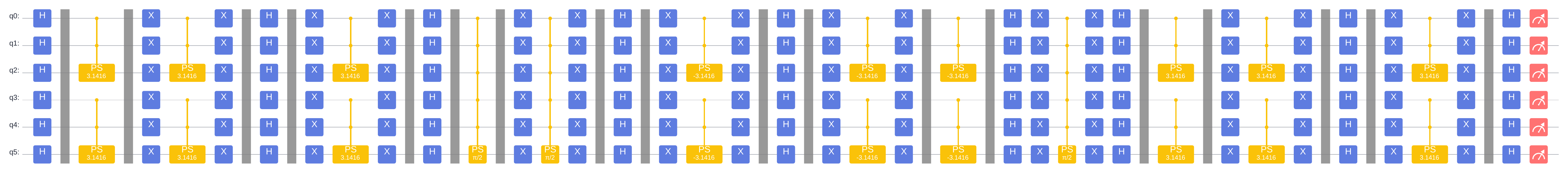}
\caption{Quantum circuit of 6-qubit DEGGA with $t=2$ computing nodes (target string $x \in \{000000,111111 \}$). The parameter in the circuit is $\phi_{n_j}=2\arcsin\left(\sin\left(\frac{\pi}{4J_{n_j}+6}\right) / \sin \left(\theta_{n_j}\right) \right) =3.1415926237874707\approx 3.1416$, where $J_{n_j}=\lfloor(\pi/2-\theta_{n_j})/(2\theta_{n_j})\rfloor$, $\theta_{n_j}= \arcsin \left({\sqrt{1/2^3}}\right)$, and $j\in\{0,1\}$. Similarly, $\phi_{f'}=2\arcsin\left( \sin\left(\frac{\pi}{4J_{f'}+6}\right)/\sin \left(\theta_{f'}\right) \right)=1.5707963267948961\approx \pi/2$, where $J_{f'}=\left \lfloor (\pi/2-\theta_{f'}) / (2\theta_{f'}) \right\rfloor,$ and $ \theta_{f'}= \arcsin \left( \sqrt{2/(2 * 2)} \right)$.}
\label{6qubit-degga2-cir}
\end{figure}

By sampling 10,000 times of the circuit depicted in Figure \ref{6qubit-degga2-cir}, the sampling outcomes are presented in Figure \ref{6qubit-degga2-out}. To ensure the reproducibility of our experimental results, we have fixed the random seed value at 43.

\begin{figure}[H]
\centering
\includegraphics[width=0.6\textwidth]{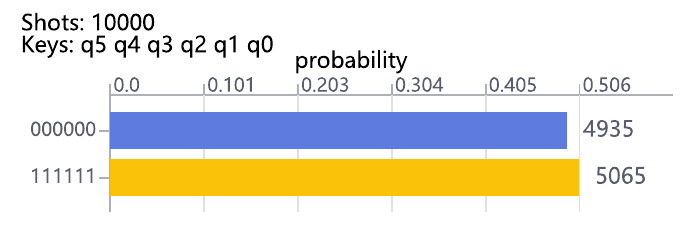}
\caption{The sampling results of 6-qubit DEGGA with $t=2$ computing nodes (target string $x \in \{000000,111111 \}$). }
\label{6qubit-degga2-out}
\end{figure}

The sampling outcomes corroborate that the measurement precisely yields $x$ within the set $\{000000, 111111\}$. Furthermore, the entire quantum circuit consists of a total of 171 quantum gates, including 3 $\text{C}^5\text{PS}$ gates and 18 $\text{C}^2\text{PS}$ gates, and has a circuit depth of 37. Despite the augmentation in quantum gates, there has been a concomitant reduction in the employment of $\text{C}^5\text{PS}$ gates, from 12 to 3.

As previously indicated, $R_ {f '}$ can be chosen to execute quantum gates between computing nodes according to actual situations. Therefore, we can construct an optimized quantum circuit with a 6-qubit DEGA with $t=2$ computational nodes (see Figure \ref{6qubit-degga2opti-cir}).

\begin{figure}[H]
\centering
\includegraphics[width=0.95\textwidth]{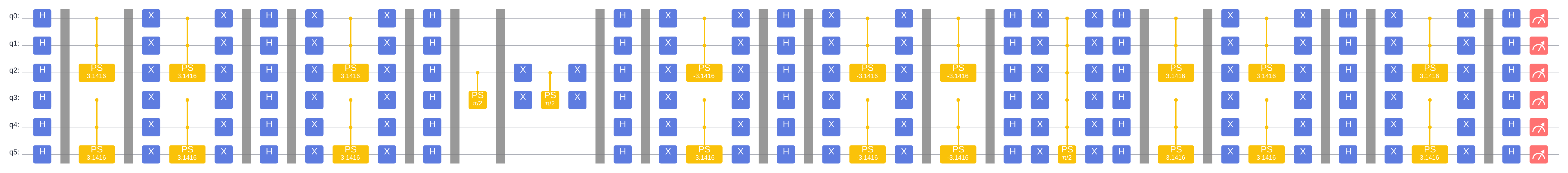}
\caption{Optimized quantum circuit of 6-qubit DEGGA with $t=2$ computing nodes (target string $x \in \{000000,111111 \}$). The parameter in the circuit is $\phi_{n_j} =3.1415926237874707\approx 3.1416$ and $\phi_{f'}=1.5707963267948961\approx \pi/2$, where $j\in\{0,1\}$. }
\label{6qubit-degga2opti-cir}
\end{figure}

By sampling 10,000 times of the circuit depicted in Figure \ref{6qubit-degga2opti-cir}, the sampling outcomes are presented in Figure \ref{6qubit-degga2opti-out}. To ensure the reproducibility of our experimental results, we have fixed the random seed value at 44.

\begin{figure}[H]
\centering
\includegraphics[width=0.6\textwidth]{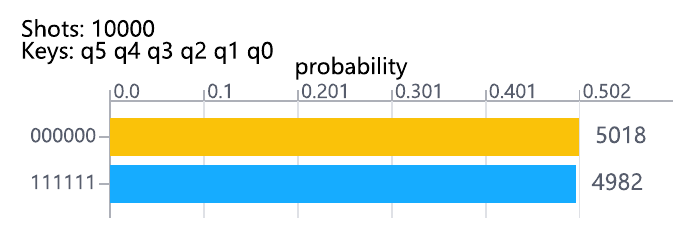}
\caption{The sampling results of optimized 6-qubit DEGGA with $t=2$ computing nodes (target string $x \in \{000000,111111 \}$). }
\label{6qubit-degga2opti-out}
\end{figure}

The sampling outcomes corroborate that the measurement precisely yields $x$ within the set $\{000000, 111111\}$. Furthermore, the entire quantum circuit consists of a total of 163 quantum gates, including 1 $\text{C}^5\text{PS}$ gate, 18 $\text{C}^2\text{PS}$ gates, and 2 $\text{C}\text{PS}$ gates, and has a circuit depth of 37. It is evident that the optimized quantum circuit streamlines the usage of quantum gates, notably curtailing the deployment of $\text{C}^5\text{PS}$ gates from 3 to 1.

\subsection{The 6-qubit DEGGA  with $t =3$ computing nodes, target string $x \in \{000000,111111 \}$}\label{6qdegga3}
At last, we will present a specific implementation of DEGGA utilizing three computing nodes. In this instance, we configure the computing nodes $t=3$, and assign an equal number of qubits to each computing node, setting $n_0 = n_1 =n_2=2$.

Firstly, we choose last $4$ bits in $x$ to divide $f$, then we can get $2^4=16$ subfunctions $f_{{0},j}: \{0,1\}^2 \rightarrow \{0,1\}$:
\begin{eqnarray}
f_{{0},0}(m_0)&=f(  m_0 0000), f_{{0},8}(m_0)&=f(  m_0 1000),\\
f_{{0},1}(m_0)&=f(  m_0 0001), f_{{0},9}(m_0)&=f(  m_0 1001),\\
f_{{0},2}(m_0)&=f(  m_0 0010), f_{{0},10}(m_0)&=f(  m_0 1010),\\
f_{{0},3}(m_0)&=f(  m_0 0011), f_{{0},11}(m_0)&=f(  m_0 1011),\\
f_{{0},4}(m_0)&=f(  m_0 0100), f_{{0},12}(m_0)&=f(  m_0 1100),\\
f_{{0},5}(m_0)&=f(  m_0 0101), f_{{0},13}(m_0)&=f(  m_0 1101),\\
f_{{0},6}(m_0)&=f(  m_0 0110), f_{{0},14}(m_0)&=f(  m_0 1110),\\
f_{{0},7}(m_0)&=f(  m_0 0111), f_{{0},15}(m_0)&=f(  m_0 1111),
\end{eqnarray}
where $m_0\in\{0,1\}^{2}$ and $j\in\{0,1, \cdots,15\}$. Obviously, $f_{{0},1}(m_0)=f_{{0},2}(m_0)=\cdots=f_{{0},14}(m_0)\equiv0$, and
\begin{eqnarray}
f_{{0},0}(m_0)&=&
\begin{cases}
1,m_0=00,\\
0,m_0\neq00,
\end{cases}\\
f_{{0},15}(m_0)&=&
\begin{cases}
1,m_0=11,\\
0,m_0\neq11,
\end{cases}
\end{eqnarray}
where $m_0\in\{0,1\}^2$.

Afterwards, we generate a new function $g_0: \{0,1\}^2 \rightarrow \{0,1\}$ through above sixteen subfunctions,
\begin{eqnarray}
g_{0}(m_0)=\text{OR} \left(f_{{0},0}(m_0), f_{{0},1}(m_0), \cdots , f_{{0},15}(m_0)\right)=
\begin{cases}
1,m_0=00,11,\\
0, \text{otherwise},
\end{cases}
\end{eqnarray}
where $m_0\in\{0,1\}^{2}$.

Secnodly, we choose first $2$ bits and last $2$ bits in $x$ to divide $f$, then we can get $2^4=16$ subfunctions $f_{{1},j}: \{0,1\}^2 \rightarrow \{0,1\}$:
\begin{eqnarray}
f_{{1},0}(m_1)&=f(  00 m_1 00), f_{{1},8}(m_1)&=f( 10 m_1 00),\\
f_{{1},1}(m_1)&=f( 00 m_1 01), f_{{1},9}(m_1)&=f( 10 m_1 01),\\
f_{{1},2}(m_1)&=f( 00 m_1 10), f_{{1},10}(m_1)&=f( 10 m_110),\\
f_{{1},3}(m_1)&=f( 00 m_1 11), f_{{1},11}(m_1)&=f( 10 m_1 11),\\
f_{{1},4}(m_1)&=f( 01 m_1 00), f_{{1},12}(m_1)&=f( 11 m_1 00),\\
f_{{1},5}(m_1)&=f( 01 m_1 01), f_{{1},13}(m_1)&=f( 11 m_1 01),\\
f_{{1},6}(m_1)&=f( 01 m_1 10), f_{{1},14}(m_1)&=f( 11 m_1 10),\\
f_{{1},7}(m_1)&=f( 01 m_1 11), f_{{1},15}(m_1)&=f( 11 m_1 11),
\end{eqnarray}
where $m_1\in\{0,1\}^{2}$ and $j\in\{0,1, \cdots,15\}$. Obviously, $f_{{1},1}(m_1)=f_{{1},2}(m_1)=\cdots=f_{{1},14}(m_1)\equiv0$, and
\begin{eqnarray}
f_{{1},0}(m_1)&=&
\begin{cases}
1,m_1=00,\\
0,m_1\neq00,
\end{cases}\\
f_{{1},15}(m_1)&=&
\begin{cases}
1,m_1=11,\\
0,m_1\neq11,
\end{cases}
\end{eqnarray}
where $m_1\in\{0,1\}^2$.

Afterwards, we generate a new function $g_1: \{0,1\}^2 \rightarrow \{0,1\}$ through above sixteen subfunctions,
\begin{eqnarray}
g_{1}(m_1)=\text{OR} \left(f_{{1},0}(m_1), f_{{1},1}(m_1), \cdots , f_{{1},15}(m_1)\right)=
\begin{cases}
1,m_1=00,11,\\
0, \text{otherwise},
\end{cases}
\end{eqnarray}
where $m_1\in\{0,1\}^{2}$.

Finally, we choose first $4$ bits in $x$ to divide $f$, then we can get $2^4=16$ subfunctions $f_{{2},j}: \{0,1\}^2 \rightarrow \{0,1\}$:
\begin{eqnarray}
f_{{2},0}(m_2)&=f( 0000 m_2), f_{{2},8}(m_2)&=f( 1000 m_2),\\
f_{{2},1}(m_2)&=f( 0001 m_2), f_{{2},9}(m_2)&=f( 1001 m_2),\\
f_{{2},2}(m_2)&=f( 0010 m_2), f_{{2},10}(m_2)&=f( 1010 m_2),\\
f_{{2},3}(m_2)&=f( 0011 m_2), f_{{2},11}(m_2)&=f( 1011 m_2),\\
f_{{2},4}(m_2)&=f( 0100 m_2), f_{{2},12}(m_2)&=f( 1100 m_2),\\
f_{{2},5}(m_2)&=f( 0101 m_2), f_{{2},13}(m_2)&=f( 1101 m_2),\\
f_{{2},6}(m_2)&=f( 0110 m_2), f_{{2},14}(m_2)&=f( 1110 m_2),\\
f_{{2},7}(m_2)&=f( 0111 m_2), f_{{2},15}(m_2)&=f( 1111 m_2),
\end{eqnarray}
where $m_2\in\{0,1\}^{2}$ and $j\in\{0,1, \cdots,15\}$. Obviously, $f_{{2},1}(m_2)=f_{{2},2}(m_2)=\cdots=f_{{2},14}(m_2)\equiv0$, and
\begin{eqnarray}
f_{{2},0}(m_2)&=&
\begin{cases}
1,m_2=00,\\
0,m_2\neq00,
\end{cases}\\
f_{{2},15}(m_2)&=&
\begin{cases}
1,m_2=11,\\
0,m_2\neq11,
\end{cases}
\end{eqnarray}
where $m_2\in\{0,1\}^2$.

Afterwards, we generate a new function $g_2: \{0,1\}^2 \rightarrow \{0,1\}$ through above sixteen subfunctions,
\begin{eqnarray}
g_{2}(m_2)=\text{OR} \left(f_{{2},0}(m_2), f_{{2},1}(m_2), \cdots , f_{{2},15}(m_2)\right)=
\begin{cases}
1,m_2=00,11,\\
0, \text{otherwise},
\end{cases}
\end{eqnarray}
where $m_2\in\{0,1\}^{2}$.

Given our assumption that it is straightforward to acquire the Oracle for each subfunction, we can construct the complete quantum circuit of 6-qubit DEGGA with $t=3$ computing nodes (see Figure \ref{6qubit-degga3-cir}).

\begin{figure}[H]
\centering
\includegraphics[width=0.95\textwidth]{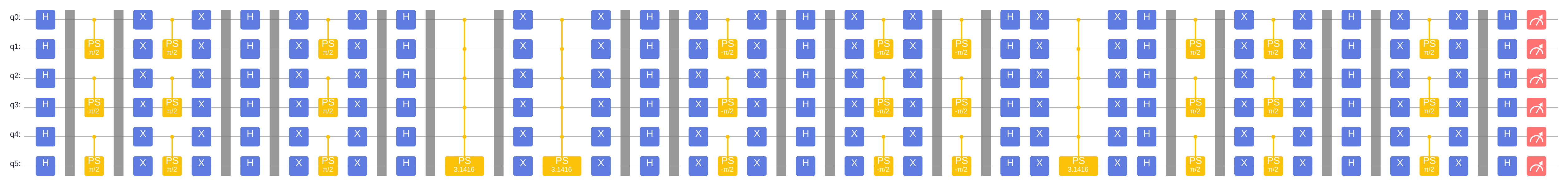}
\caption{Quantum circuit of 6-qubit DEGGA with $t=3$ computing nodes (target string $x \in \{000000,111111 \}$). The parameter in the circuit is $\phi_{n_j}=2\arcsin\left(\sin\left(\frac{\pi}{4J_{n_j}+6}\right) / \sin \left(\theta_{n_j}\right) \right)=1.5707963267948961\approx \pi/2 $, where $J_{n_j}=\lfloor(\pi/2-\theta_{n_j})/(2\theta_{n_j})\rfloor$, $\theta_{n_j}= \arcsin \left({\sqrt{1/2^2}}\right)$, and $j\in\{0,1,2\}$. Similarly, $\phi_{f'}=2\arcsin\left( \sin\left(\frac{\pi}{4J_{f'}+6}\right)/\sin \left(\theta_{f'}\right) \right) =3.1415926237874707\approx 3.1416$, where $J_{f'}=\left \lfloor (\pi/2-\theta_{f'}) / (2\theta_{f'}) \right\rfloor,$ and $ \theta_{f'}= \arcsin \left( \sqrt{2/(2 * 2 * 2)} \right)$.}
\label{6qubit-degga3-cir}
\end{figure}

By sampling 10,000 times of the circuit depicted in Figure \ref{6qubit-degga3-cir}, the sampling outcomes are presented in Figure \ref{6qubit-degga3-out}. To ensure the reproducibility of our experimental results, we have fixed the random seed value at 45.

\begin{figure}[H]
\centering
\includegraphics[width=0.6\textwidth]{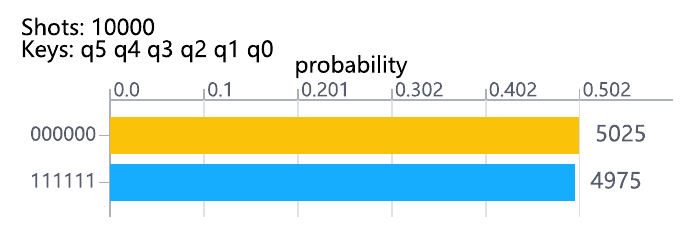}
\caption{The sampling results of 6-qubit DEGGA with $t=3$ computing nodes (target string $x \in \{000000,111111 \}$). }
\label{6qubit-degga3-out}
\end{figure}

The sampling outcomes corroborate that the measurement precisely yields $x$ within the set $\{000000, 111111\}$. Furthermore, the entire quantum circuit consists of a total of 180 quantum gates, including 3 $\text{C}^5\text{PS}$ gates and 27 $\text{C}\text{PS}$ gates, and has a circuit depth of 37. Despite the augmentation in quantum gates, there has been a concomitant reduction in the employment of $\text{C}^2\text{PS}$ gates, from 18 to 0.

Futhermore, we can construct an optimized quantum circuit with a 6-qubit DEGA with $t=3$ computational nodes (see Figure \ref{6qubit-degga3opti-cir}).

By sampling 10,000 times of the circuit depicted in Figure \ref{6qubit-degga3opti-cir}, the sampling outcomes are presented in Figure \ref{6qubit-degga3opti-out}. To ensure the reproducibility of our experimental results, we have fixed the random seed value at 46.

\begin{figure}[H]
\centering
\includegraphics[width=0.95\textwidth]{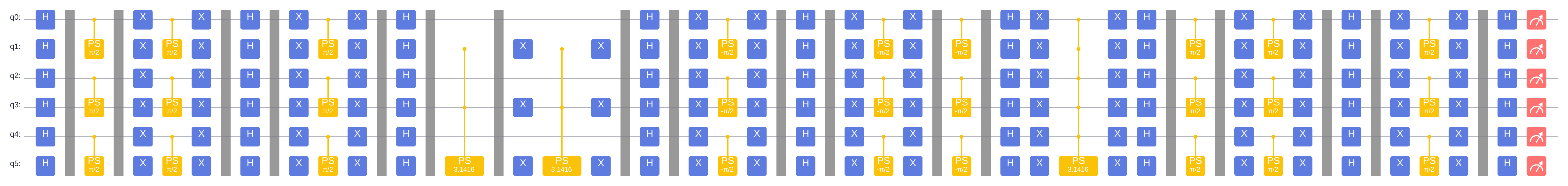}
\caption{Optimized quantum circuit of 6-qubit DEGGA with $t=3$ computing nodes (target string $x \in \{000000,111111 \}$). The parameter in the circuit is $\phi_{n_j} =1.5707963267948961\approx \pi/2$ and $\phi_{f'}=3.1415926237874707\approx 3.1416$, where $j\in\{0,1,2\}$. }
\label{6qubit-degga3opti-cir}
\end{figure}

\begin{figure}[H]
\centering
\includegraphics[width=0.6\textwidth]{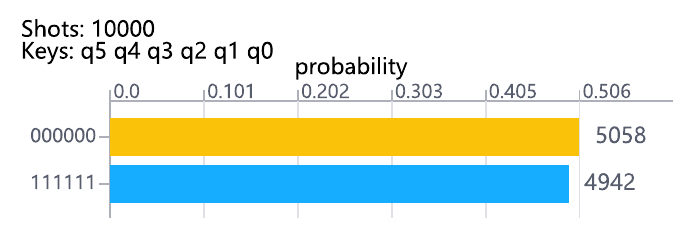}
\caption{The sampling results of optimized 6-qubit DEGGA with $t=3$ computing nodes (target string $x \in \{000000,111111 \}$). }
\label{6qubit-degga3opti-out}
\end{figure}

The sampling outcomes corroborate that the measurement precisely yields $x$ within the set $\{000000, 111111\}$. Furthermore, the entire quantum circuit consists of a total of 174 quantum gates, including 1 $\text{C}^5\text{PS}$ gate, 2 $\text{C}^2\text{PS}$ gates, and 27 $\text{C}\text{PS}$ gates, and has a circuit depth of 37. It is evident that the optimized quantum circuit streamlines the usage of quantum gates, notably curtailing the deployment of $\text{C}^2\text{PS}$ gates from 18 to 2.

Through the aforementioned experiments, we can ascertain the specific procedures involved in implementing a 6-qubit DEGGA. To summarize the outcomes of all the experiments, we present a succinct statistical table as depicted in Table \ref{simplestatistic}.

\begin{table}[H]
  \centering
  \caption{A simple statistic for the modified Grover's algorithm and DEGGA.}
\scalebox{0.7}{
  \begin{tabular}{lcccccc}
    \toprule
   &Modified Grover's algorithm & DEGGA-2 & Optimized DEGGA-2 & DEGGA-3 & Optimized DEGGA-3  \\
    \midrule
\text{1. The number of computing nodes.} &1&2&2&3&3 \\
\text{2. Quantum circuit.} &\text{Figure} \ref{6qubit-long-cir}&\text{Figure} \ref{6qubit-degga2-cir}&\text{Figure} \ref{6qubit-degga2opti-cir}&\text{Figure} \ref{6qubit-degga3-cir}&\text{Figure} \ref{6qubit-degga3opti-cir}\\
\text{3. Random seed value.}
&42&43&44&45&46 \\
\text{4. The sampling results.}&\text{Figure} \ref{6qubit-long-out}&\text{Figure} \ref{6qubit-degga2-out}&\text{Figure} \ref{6qubit-degga2opti-out}&\text{Figure} \ref{6qubit-degga3-out}&\text{Figure} \ref{6qubit-degga3opti-out}\\
\text{5. Target strings.}& $ \{000000,111111 \}$&$ \{000000,111111 \}$&$ \{000000,111111 \}$&$ \{000000,111111 \}$&$ \{000000,111111 \}$\\
\textbf{6. Does it achieve an exact search?}
&\textbf{Yes}&\textbf{Yes}&\textbf{Yes}&\textbf{Yes}&\textbf{Yes} \\
    \bottomrule
  \end{tabular}}
\label{simplestatistic}
\end{table}

It is obviously known from the table that both the modified Grover's algorithm and the DEGGA (two-node and three-node) are capable of achieving an exact search. Additionally, we summarize the quantum gates and circuit depths necessitated by the modified Grover's algorithm and the DEGGA for searching 000000 and 111111, as depicted in Figure \ref{comparetable}.

\begin{figure}[H]
\centering
\includegraphics[width=0.65\textwidth]{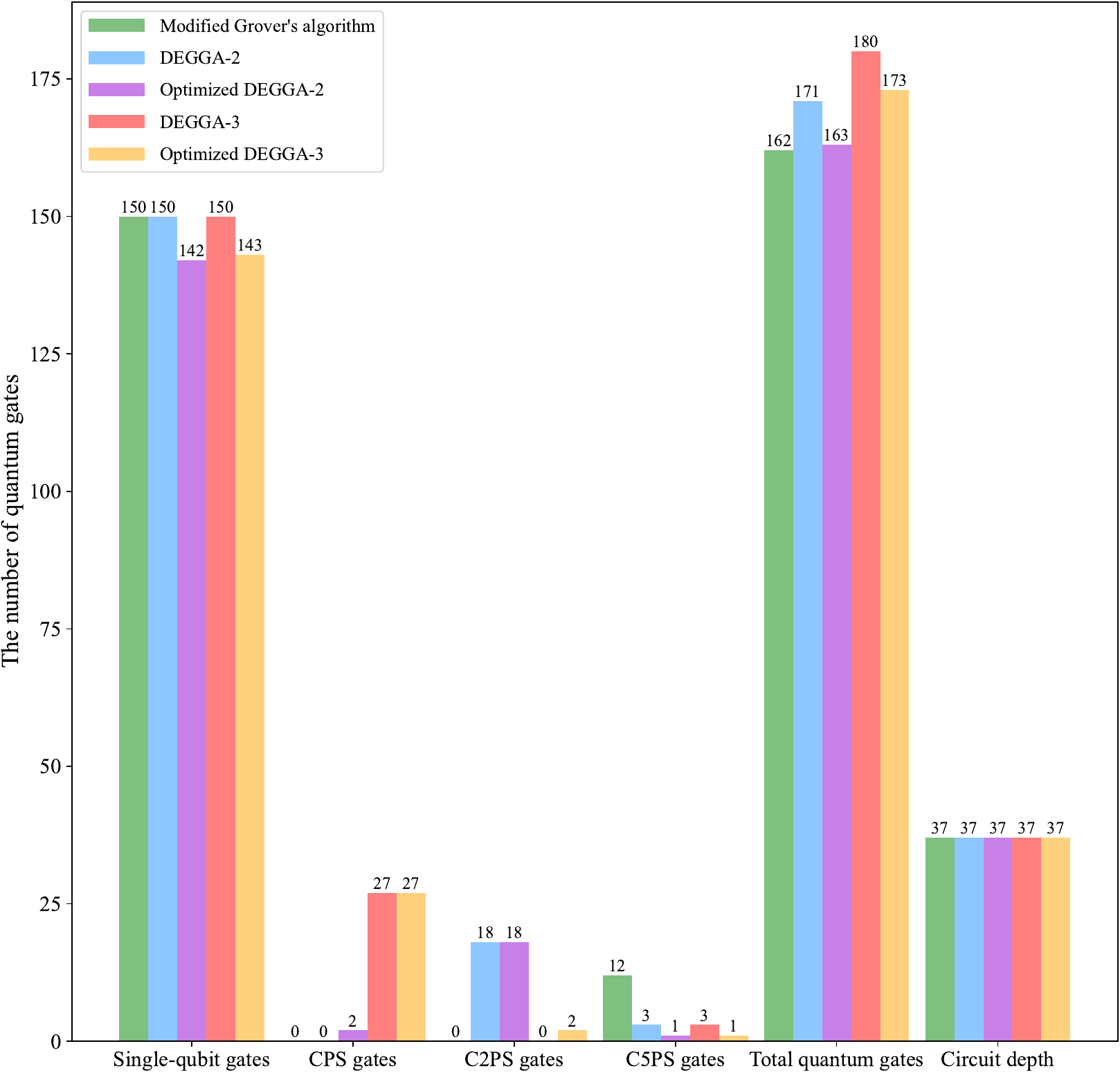}
\caption{A straightforward comparison between the modified Grover's algorithm and the DEGGA.}
\label{comparetable}
\end{figure}

From the statistical comparison presented above, we can ascertain that: (1) the circuit depth of each algorithm is 37; (2) while the modified Grover's algorithm necessitates fewer quantum gates, it requires a greater number of $\text{C}^{5}\text{PS}$ gates; (3) the optimized DEGGA demands fewer quantum gates compared to its pre-optimized counterpart, particularly diminishing the reliance on $\text{C}^{5}\text{PS}$ gates; (4) as the number of nodes escalates (from 2 to 3), there is an additional reduction in the necessary $\text{C}^{2}\text{PS}$ gates. In other words, (optimized) DEGGA-3 requires fewer multi-qubit gates compared to (optimized) DEGGA-2.

\subsection{Decomposition of $\text{C}^{5}\text{PS}$ gates}\label{sectionDecomposition}
In this subsection, we further consider the decomposition of $\text{C}^{5}\text{PS}$ gates into a combination of single-qubit and double-qubit gates. By employing such a decomposition, the disparity between the modified Grover's algorithm and DEGGA in terms of the quantum gate count and circuit depth becomes more pronounced, thereby further accentuating the unequivocal benefits of distributed quantum algorithms.

\begin{lemma}
For any single-qubit gate $U$, $\text{C}^{5}\text{U}$ gate can be decomposed into the following quantum circuit  (see Figure \ref{C5Ugate}),
\begin{figure}[H]
\centering
\includegraphics[width=0.33\textwidth]{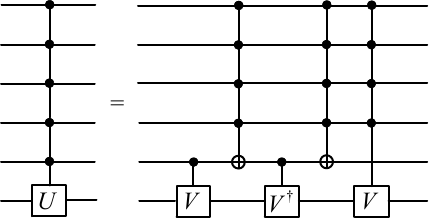}
\caption{Quantum circuit for the decomposition of $\text{C}^{5}\text{U}$ gate.}
\label{C5Ugate}
\end{figure}
where $V$ represents another single-qubit gate that fulfills the condition $V^2=U$.
\label{lemmaC5U}
\end{lemma}

\begin{proof}
For brevity, the detailed proof is omitted and can be found in  \cite{RefBarenco1995}.
\end{proof}

Specifically, if $U = \text{PS}(\theta)$, then $V = \text{PS}(\theta/2)$ and $V^{\dagger} = \text{PS}(-\theta/2)$ can be readily derived. The decomposed quantum circuit contains two $\text{C}^{4}\text{X}$ gates, two CPS gates and one $\text{C}^{4}\text{PS}$ gate. The decomposition of the $\text{C}^{4}\text{X}$ gate will be addressed subsequently.

For the $\text{C}^{4}\text{PS}$ gate, by employing a decomposition analogous to Lemma \ref{lemmaC5U}, we derive the ensuing quantum circuit (as depicted in Figure \ref{C4Ugate}), wherein $W$ denotes another single-qubit gate that satisfies the equation $W^2 = V$.

\begin{figure}[H]
\centering
\includegraphics[width=0.35\textwidth]{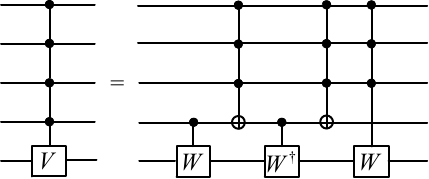}
\caption{Quantum circuit for the decomposition of $\text{C}^{4}\text{V}$ gate.}
\label{C4Ugate}
\end{figure}
In other words, if $V = \text{PS}(\theta/2)$, then $W = \text{PS}(\theta/4)$ and $W^{\dagger} = \text{PS}(-\theta/4)$ can be readily derived. The decomposed quantum circuit contains two $\text{C}^{3}\text{X}$ gates, two CPS gates and one $\text{C}^{3}\text{PS}$ gate. The decomposition of the $\text{C}^{3}\text{X}$ gate will be addressed subsequently. 

For the $\text{C}^{3}\text{PS}$ gate, it can be decomposed into the following quantum circuit (see Lemma \ref{lemmaC3U}). 

\begin{lemma}
For any single-qubit gate $W$, $\text{C}^{3}\text{W}$ gate can be decomposed into the following quantum circuit  (see Figure \ref{C3Ugate}),
\begin{figure}[H]
\centering
\includegraphics[width=0.5\textwidth]{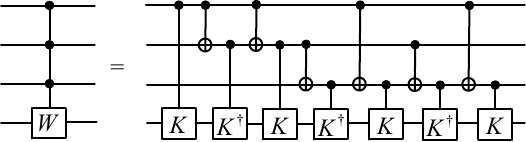}
\caption{Quantum circuit for the decomposition of $\text{C}^{3}\text{W}$ gate.}
\label{C3Ugate}
\end{figure}
where $K$ represents another single-qubit gate that fulfills the condition $K^4=W$.
\label{lemmaC3U}
\end{lemma}

\begin{proof}
For brevity, the detailed proof is omitted and can be found in  \cite{RefBarenco1995}.
\end{proof}

Specifically, if $W = \text{PS}(\theta/4)$, then $K = \text{PS}(\theta/16)$ and $K^{\dagger} = \text{PS}(-\theta/16)$ can be readily derived. The resulting decomposed quantum circuit exclusively comprises two-qubit gates, namely CNOT and CPS gates.

Thus far, with respect to the decomposition of the $\text{C}^{5}\text{PS}$ gate, it is necessary to consider the decompositions of both the $\text{C}^{4}\text{X}$ and $\text{C}^{3}\text{X}$ gates.

\begin{lemma}
For the $\text{C}^{4}\text{X}$ gate, it can be decomposed into the quantum circuit detailed below  (see Figure \ref{C4Xgate}).
\begin{figure}[H]
\centering
\includegraphics[width=0.55\textwidth]{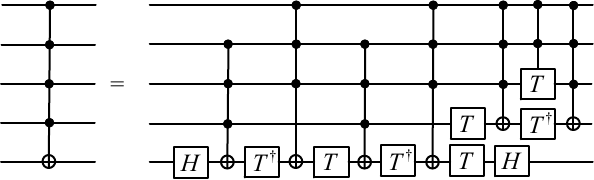}
\caption{Quantum circuit for the decomposition of $\text{C}^{4}\text{X}$ gate.}
\label{C4Xgate}
\end{figure}
\label{lemmaC4X}
\end{lemma}

\begin{proof}
It can be readily verified that the unitary matrices corresponding to the quantum circuits on either side of the equality sign are equivalent.
\end{proof}

For the $\text{C}^{2}\text{T}$ gate, it is fundamentally identical to the $\text{C}^{2}\text{PS}(\pi/4)$ gate, implying that it can be decomposed using the quantum circuit detailed below (as referenced in Lemma \ref{lemmaC2U}).

\begin{lemma}
For any single-qubit gate $U$, $\text{C}^{2}\text{U}$ gate can be decomposed into the following quantum circuit  (see Figure \ref{C2Ugate})
\begin{figure}[H]
\centering
\includegraphics[width=0.4\textwidth]{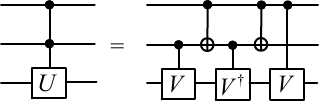}
\caption{Quantum circuit for the decomposition of $\text{C}^{2}\text{U}$ gate.}
\label{C2Ugate}
\end{figure}
where $V$ represents another single-qubit gate that fulfills the condition $V^2=U$.
\label{lemmaC2U}
\end{lemma}

\begin{proof}
For brevity, the detailed proof is omitted and can be found in  \cite{RefBarenco1995}.
\end{proof}

Specifically, if $U =\text{PS}(\pi/4)$, then $V = \text{PS}(\pi/8)$ and $V^{\dagger} = \text{PS}(-\pi/8)$ can be readily derived. The resulting decomposed quantum circuit exclusively comprises two-qubit gates, namely CNOT and CPS gates. 

In order to decompose $\text{C}^{3}\text{X}$, we give the following Lemma \ref{lemmaC3X}.
\begin{lemma}
For the $\text{C}^{3}\text{X}$ gate, it can be decomposed into the quantum circuit detailed below  (see Figure \ref{C3Xgate}).
\begin{figure}[H]
\centering
\includegraphics[width=0.6\textwidth]{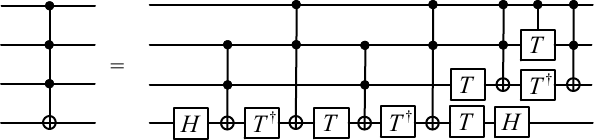}
\caption{Quantum circuit for the decomposition of $\text{C}^{3}\text{X}$ gate.}
\label{C3Xgate}
\end{figure}
\label{lemmaC3X}
\end{lemma}

\begin{proof}
It can be readily verified that the unitary matrices corresponding to the quantum circuits on either side of the equality sign are equivalent. Specifically, the $\text{C}\text{T}$ gate is equivalent to the $\text{C}\text{PS}(\pi/4)$ gate. 
\end{proof}

Lastly, we present the decomposition for the $\text{C}^{2}\text{X}$ gate (Toffoli gate).

\begin{lemma}
For the $\text{C}^{2}\text{X}$ gate, it can be decomposed into the quantum circuit detailed below (see Figure \ref{C2Xgate}).
\begin{figure}[H]
\centering
\includegraphics[width=0.65\textwidth]{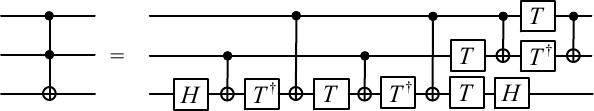}
\caption{Quantum circuit for the decomposition of $\text{C}^{2}\text{X}$ gate (Toffoli gate).}
\label{C2Xgate}
\end{figure}
\label{lemmaC2X}
\end{lemma}

\begin{proof}
For brevity, the detailed proof is omitted and can be found in  \cite{RefNielsen2002}.
\end{proof}

In summary, Table \ref{simplestatisticquantumgates} presents a compilation of the necessary quantum gates count for the $\text{C}^{5}\text{PS}$ and other gates along with the circuit depths of the decomposed quantum circuits. It can be seen that a total of 1429 quantum gates are required to decompose a $\text{C}^{5}\text{PS}$ gate, with 884 single-qubit gates and 545 double-qubit gates (514 for CNOT gates 14 for CT gates and 17 for CPS gates), for a total circuit depth of 959. 

Note that the statistics do not take into account the circuit optimization process. Furtuermore, there might exist more optimal decomposition methods. However, an exhaustive discussion of these methods is not within the purview of the current paper.

\begin{table}[H]
  \centering
  \caption{A statistic for the decomposition of $\text{C}^{5}\text{PS}$ and other gates (ignore the process of circuit optimisation).}
\scalebox{0.7}{
  \begin{tabular}{lcccccccc}
    \toprule
   &single-qubit gates & double-qubit gates & CNOT gates &CT gate& CPS gates & total quantum gates & circuit depth  \\
    \midrule
\text{1. $\text{C}^{2}\text{X}$ gate}
&9&6&6&0&0&15&12 \\
\text{2. $\text{C}^{3}\text{X}$ gate}
&62&37&36&1&0&99&69 \\
\text{3. $\text{C}^{4}\text{X}$ gate}
&380&227&218&6&3&607&408 \\
\textbf{4. $\text{C}^{2}\text{PS}$ gate} &\textbf{0}&\textbf{5}&\textbf{2}&\textbf{0}&\textbf{3}&\textbf{5}&\textbf{5} \\
\text{5. $\text{C}^{3}\text{PS}$ gate} &0&13&6&0&7&13&13 \\
\text{6. $\text{C}^{4}\text{PS}$ gate} &124&89&78&2&9&213&149 \\
\textbf{7. $\text{C}^{5}\text{PS}$ gate}
&\textbf{884}&\textbf{545}&\textbf{514}&\textbf{14}&\textbf{17}&\textbf{1429}&\textbf{959} \\
    \bottomrule
  \end{tabular}}
\label{simplestatisticquantumgates}
\end{table}

Last but not least, we substitute the above decomposition into the modified Grover's algorithm and the DEGGA for searching 000000 and 111111, and again summarise their required quantum gates as well as circuit depths. The comparative results after the decomposition are shown in Figure \ref{comparetabledecom}.

\begin{figure}[H]
\centering
\includegraphics[width=0.65\textwidth]{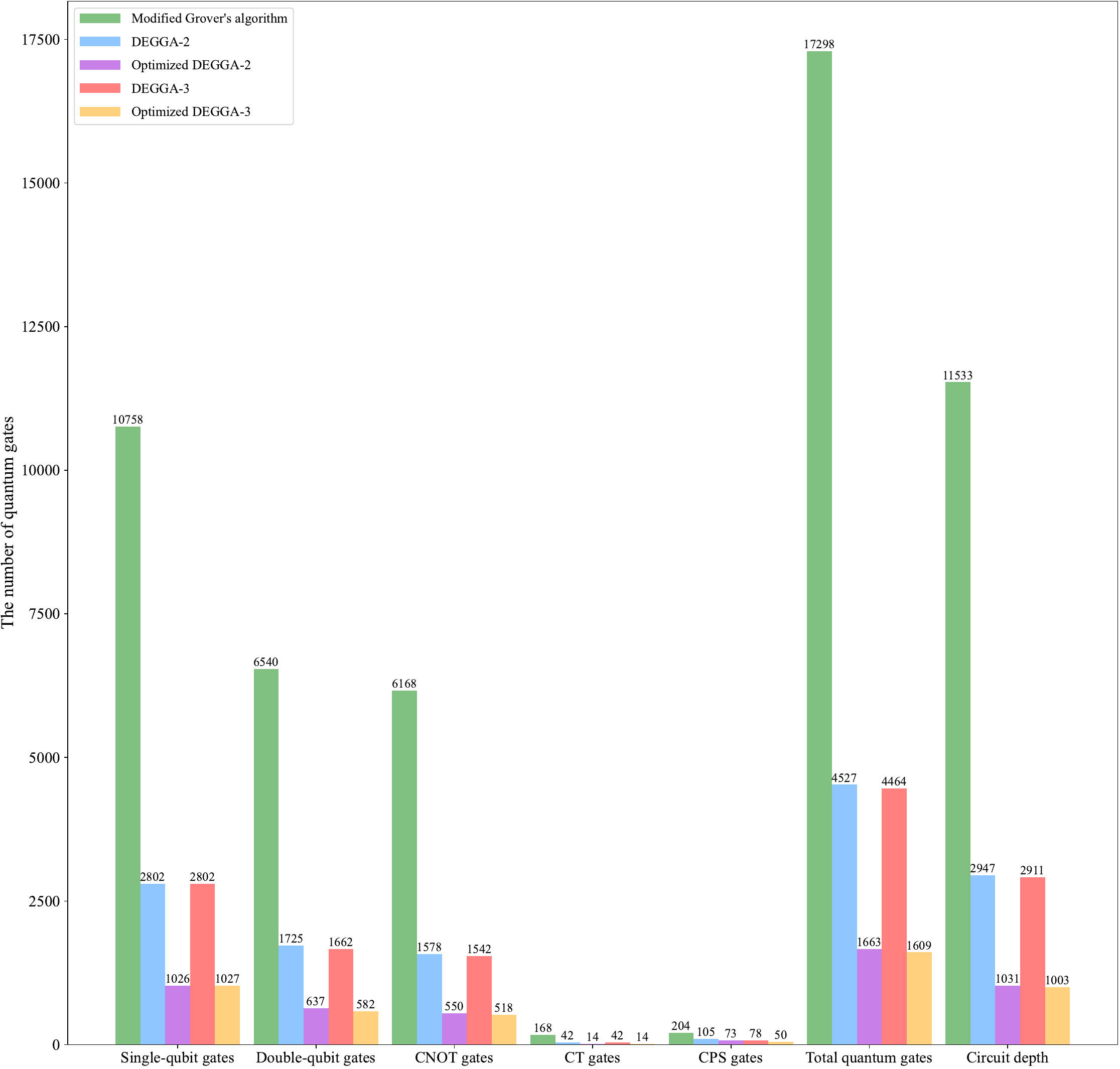}
\caption{A straightforward comparison between the modified Grover's algorithm and the DEGGA after decomposion (ignore the process of circuit optimisation).}
\label{comparetabledecom}
\end{figure}

%From the above comparison, it can be seen that DEGGA requires fewer quantum gates and shallower circuit depths than the modified Grover's algorithm, thus rendering our proposed distributed quantum algorithms more practical in the NISQ era.

The comparison elucidated above reveals that DEGGA requires a reduced number of quantum gates and shallower circuit depths in contrast to the modified Grover's algorithm. Specifically, a $90.7\%$ reduction in quantum gates and a $91.3\%$ diminished circuit depth, thereby rendering our proposed distributed quantum algorithm more feasible for application in the NISQ era.

\section{Conclusion}\label{conclusion}
Distributed quantum computation has been identified as a promising and advantageous application in the noisy intermediate-scale quantum (NISQ) era, in which each computational node necessitates fewer qubits and quantum gates. In other words, distributed quantum computing holds the potential to solve problems with greater efficiency.

In this study, we have focused on a generalized search problem involving multiple targets within an unordered database. Subsequently, we have developed a Distributed Exact Generalized Grover's Algorithm (DEGGA) that tackles the initial challenge of generalized search by decomposing it into arbitrary $t$ components, where $2 \leq t \leq n$.

More specifically, (1) our algorithm is accurate, ensuring a theoretical probability of $100\%$ for identifying the target states; (2) if the number of targets is fixed, the pivotal factor influencing the circuit depth of DEGGA is the partitioning strategy, rather than the magnitude of $n$. Conversely, the modified Grover's algorithm augments the circuit depth as $n$ escalates; (3) our approach only requires $n$ qubits, devoid of any auxiliary qubits.

Ultimately, we have elucidated the resolutions (two-node and three-node) of a particular generalized search issue incorporating two goal strings (000000 and 111111) based on the DEGGA. Furthermore, we have outlined the detailed procedures involved in implementing a 6-qubit DEGGA on MindSpore Quantum (a quantum simulation software). Eventually, through the decomposition of multi-qubit gates, DEGGA diminishes the utilization of quantum gates by $90.7\%$ and decreases the circuit depth by $91.3\%$ in comparison to the modified Grover's algorithm by Long. It is increasingly evident that distributed quantum algorithms offer augmented practicality in the NISQ era. Our findings not only demonstrate the practicality and efficiency of our proposed approach but also provide a solid foundation for future research in this domain. Besides, it is still an open question how to definitively ascertain the number of targets of a Boolean function (or subfunction) by a quantum algorithm.

\section*{Declaration of competing interest}
The authors declare that they possess no conflicting financial interests or personal relationships that could potentially bias the research findings presented in this paper.

\section*{Data availability statement}
The manuscript does not contain any accompanying data.

\section*{Acknowledgements}
\noindent This work is supported by the China Postdoctoral Science Foundation under Grant No.2023M740874, the Guangdong Provincial Quantum Science Strategic Initiative under Grant No.GDZX2303003 and No.GDZX2303007, the Fundamental Research Funds for the Central Universities, Sun Yat-sen University under Grant No.2021qntd28, the Fundamental Research Funds for the Central Universities, Sun Yat-sen University under Grant No.2023lgbj020, the SYSU Key Project of Advanced Research, the Shenzhen Science and Technology Program under Grant No. JCYJ20220818102003006, the Shenzhen Science and Technology Program under Grant No.2021Szvup172, and the Innovation Program for Quantum Science and Technology under Grant No.2021ZD0302901. The authors thank Qudoor Technology for the supports from its trapped-ion quantum computing platform.

%the National Natural Science Foundation of China under Grant No. xxxxxxxx, the Natural Science Foundation of Guangdong Province of China under Grant No. xxxxxxxx, 

\balance

\end{document}